\documentclass{eptcs}
\usepackage{amsmath}
\usepackage{amsthm}
\usepackage{amsfonts}
\usepackage{amssymb}
\usepackage{thmtools}
\usepackage{hyperref}
\usepackage{float}
\usepackage{textcomp}
\usepackage{tikz}
\usepackage{comment}
\usepackage{lipsum}
\usepackage{autonum}
\usepackage{soul}
\usepackage{appendix}
\usepackage{subcaption}
\usepackage{graphicx}
\usetikzlibrary{positioning, calc, shapes, decorations.pathreplacing, arrows.meta}
\usepackage{amsthm}
\usepackage{amsthm}
\usepackage{hyperref}
\usepackage[shortlabels]{enumitem}
\newtheorem{theorem}{Theorem}[section]
\newtheorem{corollary}[theorem]{Corollary}
\newtheorem{lemma}[theorem]{Lemma}
\newtheorem{proposition}[theorem]{Proposition}

\theoremstyle{remark}
\newtheorem{remark}[theorem]{Remark}

\theoremstyle{definition}
\newtheorem{definition}[theorem]{Definition}


\providecommand{\tightlist}{%
  \setlength{\itemsep}{0pt}\setlength{\parskip}{0pt}}
\usetikzlibrary{decorations.pathmorphing}

\DeclareMathOperator{\Stab}{Stab}

\DeclareMathOperator{\PGL}{PGL}
\DeclareMathOperator{\PU}{PU}
\DeclareMathOperator{\GL}{GL}

\DeclareMathOperator{\diag}{diag}
\DeclareMathOperator{\U}{U}
\DeclareMathOperator{\M}{M}

\newcommand{\dual}[1]{{#1}^{\sharp}}
\usepackage{iftex}

\ifpdf
  \usepackage[T1]{fontenc}        % Recommended with pdflatex
\else
  \usepackage{breakurl}           % Not needed if you use pdflatex only.
\fi

\title{Buildings for Synthesis with Clifford+R}
\author{
Mark Deaconu
\institute{Institute for Quantum Computing\\
University of Waterloo}
\and
Nihar Gargava
\institute{Institut de Mathématiques d’Orsay\\
Université Paris-Saclay}
\and
Amolak Ratan Kalra
\institute{Institute for Quantum Computing\\
David R. Cheriton School of Computer Science\\
University of Waterloo\\
Perimeter Institute for Theoretical Physics\\
}
\and
Michele Mosca
\institute{Institute for Quantum Computing\\
Dept. of Combinatorics and Optimization\\
University of Waterloo\\
Perimeter Institute for Theoretical Physics\\
}
\and\and
Jon Yard
\institute{QuScript Inc.\\}
}
\def\titlerunning{Buildings for Synthesis with Clifford+R}
\def\authorrunning{M Deaconu, N.Gargava, A.R Kalra, M.Mosca \& J.Yard }
\begin{document}
\maketitle

\begin{abstract}
We study the problem of exact synthesis for the Clifford$+\mathsf{R}$ gate set and give the explicit structure of the underlying Bruhat-Tits building for this group. In this process, we also give an alternative proof of the arithmetic nature of this gate set.
\end{abstract}

\section{General synthesis problem} 

Let $\PU(d)$ be the group of unitary matrices up to phases and let $d(\cdot,\cdot)$ be some natural distance function on $\PU(d)$.
The general circuit synthesis problem can be stated as follows: Given a target unitary matrix $U \in \PU(d)$ and a universal gate set $G = \{g_1, \cdots, g_n \} \subseteq \PU(d)$, output a circuit consisting of words over $G$ that realizes $U$. A good synthesis algorithm is one in which, given a target error $\varepsilon > 0$ and $U\in \PU(d)$, the algorithm outputs a circuit $\gamma  = g_{i_1} g_{i_2} \cdots g_{i_r}$ for some $r$ that is at most a polynomial in $\log \tfrac{1}{\varepsilon}$.

In the special case where an exact decomposition of the unitary $U$ over $G$ is possible, the problem is called exact synthesis. One method to address this problem is to give an explicit algebraic characterization of those unitaries that are exactly implementable using $G$. This boils down to proving that the group generated by $G$ consists of matrices having entries in a particular ring $R \subseteq \mathbb{C}$. For example, \cite{kmm} shows an exact synthesis characterization for the single qubit Clifford$+T$ gate set. To do so, they show that the group generated by this gate set is $PU_{2}( \smash{ \mathbb{Z}[\frac{1}{\sqrt{2}}},i])$. This approach of giving an algebraic characterization for the group generated by a gate set $G$, can also be used to give more efficient algorithms to solve the general circuit synthesis problem (approximate synthesis) via a ``rounding-off" procedure, see \cite{ross2014optimal, ross2,yardapprox,Kliuchnikov2023shorterquantum}.
\subsection{Our work}

In this work, we will focus on the exact synthesis problem for single qutrits ($d = 3$). Previously, algebraic characterizations similar to the qubit case have been shown for certain qutrit gate sets, see \cite{Kalra2025synthesisarithmetic, evra2024arithmeticitycoveringrate9cyclotomic,kalra2024multiqutritexactsynthesis,Glaudell_2024,prakash2018normal}. In particular, \cite{evra2024arithmeticitycoveringrate9cyclotomic} studies both the exact and approximate synthesis problem for the Clifford$+\mathcal{D}$ gate set using Bruhat-Tits buildings.

Motivated by these results, and with the goal of furthering the connections between Bruhat-Tits theory and circuit synthesis, we study the problem of synthesis over the qutrit Clifford$+\mathsf{R}$ gate set. This gate set was first studied in \cite{anwar2012qutrit}, where a protocol to implement it in a fault tolerant manner via magic state distillation was proposed. Following this, the Clifford$+\mathsf{R}$ gate set has been studied in different contexts (see \cite{cui2015universal,yardapprox,Bocharov_2017,bocharov2016noteoptimalityquantumcircuits,glaudell_et_al:LIPIcs.TQC.2022.12}). Some progress was made on the problem of synthesis using this gate set soon after its proposal in \cite{vadymanyons} and more recently, in \cite{Kalra2025synthesisarithmetic,gustafson2025synthesis}.

In this paper, we give the structure of the Bruhat-Tits building associated with the Clifford$+\mathsf{R}$ gate set. This allows us to give a new proof for the arithmeticity of this gate set, which was first proven in \cite{Kalra2025synthesisarithmetic}. 

In this paper, we give an entirely different proof by constructing the Bruhat-Tits building for $U_{3}\left(\mathbb{Z}[\chi^{-1}]\right)$ and show that this building is in fact a tree. One advantage of this point-of-view is that the circuit synthesis algorithm can be visualized in terms of a path traversal on such a tree.
\subsection{Bruhat-Tits Theory and Gate Synthesis}
In number theory, the statement ``Clifford$+T$ satisfies a ring equality'' would usually be stated as ``the Clifford$+T$ gate set generates an S-arithmetic group''. Such arithmetic groups are a well studied classical object
(see \cite{borel1989finiteness, prasad1989volumes}). 
One of the main tools to understand the properties of these groups (and hence, exact synthesis) is to study their action on certain combinatorial structures (simplicial complexes) known as Bruhat-Tits buildings. This interplay of number theory and gate synthesis has been explored for various arithmetic gate sets \cite{kliuchnikov2024multiqubitcircuitsynthesishermitian, evra2024arithmeticitycoveringrate9cyclotomic, parzanchevski2018super,blackman2023fast,yardapprox}.

The other aim of this paper is to further this connection and give a relatively self-contained exposition of these mathematical tools (Bruhat-Tits theory) used to establish synthesis results of this nature. In particular, we present an exposition on the construction of the Bruhat-Tits building for $PU(3)$ while avoiding traditional definitions involving root spaces of algebraic tori inside p-adic Lie groups. In the past, Bruhat-Tits theory has been useful in showing that when finitely generated subgroups are not arithmetic, they are thin \cite{evra2024arithmeticitycoveringrate9cyclotomic}. 
Due to the lack of an arithmetic structure, thin groups are more difficult to deal with as one does not have an efficiently computable membership criterion for thin groups in general \cite{sarnak2012notes}. For the circuit synthesis problem, the implication is that if the gate set $G$ in question is thin, it does not satisfy a ring equality. This makes the problem of finding good synthesis algorithms for the gate set $G$ much harder. So in summary, thinness is undesirable and arithmeticity is desirable for finitely generated groups of quantum gates.

For the case of $S$-arithmetic unitary groups that the synthesis problem is concerned with, one is greatly aided by the use of the concrete lattice chain models of \cite{abramenko2002lattice} that distill the mathematical content of the general theory into a tractable simplicial complex. This makes the sophisticated number-theoretic machinery more amenable to explicit calculations over local rings and their finite residue fields. This way, we hope that our paper helps to bridge the gap between mathematicians and quantum computing specialists interested in this topic.
\section{Preliminaries}
Denote  $F = \mathbb{Q}(\omega):=\{a+b\omega ~|~a,b\in \mathbb{Q}\}$ where $\omega=e^{\frac{2\pi i}{3}}$ is the primitive third root of unity.
This is a number field of degree $[F:\mathbb{Q}]=2$ over $\mathbb{Q}$. The ring of integers of $F$ is $\mathcal{O}_F = \mathbb{Z}[\omega]:=\{a+b\omega~|~a,b\in \mathbb{Z}\}$. 
Here, the following is true:
\begin{equation}
  3 \mathcal{O}_F = \chi^{2} \mathcal{O}_F,
  \text{ where }\chi = 1- \omega.
\end{equation}
$\chi$ generates a principal prime ideal in $\mathcal{O}_F$ and we will denote $\pi = \chi  \mathcal{O}_{F}$. In particular, we have that the residue field $\mathcal{O}_F/\chi \mathcal{O}_F$ is $\mathbb{F}_3$. \footnote{$f: \mathcal{O}_{F} \rightarrow \mathbb{F}_3^2$, $f(a+b\omega) = (\bar{a},\bar{b})$ where $\bar{}$ denotes mod 3. $ker(f) = 3\mathcal{O}_F$. By the first isomorphism theorem, $\mathcal{O}_F/3\mathcal{O}_F \cong \mathbb{F}_3^2$ as an additive group. Since $\chi$ is a principal prime ideal $\mathcal{O}_F/\chi \mathcal{O}_F$ will be a field. It would have to have 3 elements because $\mathcal{O}_F/\chi^2 \mathcal{O}_F$ has 9 elements.}

\subsection{\texorpdfstring{The Clifford$+\mathsf{R}$ gate set}{The Clifford+R gate set}}
\begin{definition}
\label{de:hsr}
The Clifford$+\mathsf{R}$ gate set for qutrits is a group generated by the following matrices:
\[
H=\frac{i}{\sqrt{3}}\begin{pmatrix}
1 & 1 & 1\\
1 & \omega & \omega^{2}\\
1 & \omega^{2} & \omega\\
\end{pmatrix}~~~
S=\begin{pmatrix}
1 & 0 & 0\\
0 & \omega & 0\\
0 & 0 & 1\\
\end{pmatrix}~~~
R=
\begin{pmatrix}
1 & 0 & 0\\
0 & 1 & 0\\
0 & 0 & -1\\
\end{pmatrix}
\]
\end{definition}
In \cite{Kalra2025synthesisarithmetic}, the following theorem was proved:  
\begin{theorem}[Corollary 5.8 in \cite{Kalra2025synthesisarithmetic}] 
\label{th:ring_equality_p_is_3}
The group generated by the Clifford$+\mathsf{R}$ gate set is
\[
\langle H,S,R \rangle  =\U_{3}\left(\mathbb{Z}[\chi^{-1}]\right)
\]
\end{theorem}
Here the group on the right is the set of unitaries with entries in $\mathbb{Z}[\chi^{-1}]$ and $\mathbb{Z}[\tfrac{1}{3}, \omega]= \mathbb{Z}[\chi^{-1}]$.
\subsection{Local fields}

For $F$, we define $F_{\pi}$ to be the local completion of $F$ with respect to the prime ideal $ \pi = \chi \mathcal{O}_{F}$. This is defined as the metric completion of $F$ 
under the $\pi$-adic absolute value\footnote{It is also possible to define the absolute value as $|x|_{\pi} = e^{-v_{\pi}(x)}$. 
The choice presented here is used to ensure that $|3|_{\pi} = |3|_{3} = \tfrac{1}{3}$ as desired, where $|\ |_{p}$ is the usual $p$-adic norm on $\mathbb{Q}$.}
given by 
\begin{align}
	| \ |_{\pi} :F & \to \mathbb{R}_{\geq 0} , \quad
	x \mapsto |x|_{\pi} = 
	\begin{cases}
	3^{-\frac{{v_{\pi}(x)}}{2}} 
	& \text{if }  x \neq 0\\
	 0 & \text{if }x = 0
	\end{cases}
	. 
\end{align}
Here we define $v_{\pi}(x)$ for $x \in \mathcal{O}_{F} \setminus \{0\}$ as $
v_{\pi}(x) = \max\{ n \in \mathbb{Z} \mid \chi^{-n}x \in \mathcal{O}_{F}\},$
and for $x \in F \setminus \{0\}$ as
\begin{equation}
    v_{\pi}(x) = v_\pi(a) - v_{\pi}(b), \text{where }x=a/b \text{ for some }a,b \in \mathcal{O}_{F}.
\end{equation}

\begin{remark}
\label{re:sde}
The $v_{\pi}$ defined here is called the $\pi$-adic valuation in number theory. In some quantum circuit synthesis literature \cite{Kalra2025synthesisarithmetic,kmm}, this quantity is known as the smallest denominator exponent.
\end{remark}

We define $\mathcal{O}_{\pi} \subseteq F_{\pi}$ to be the local completion of $\mathcal{O}_F$ with respect to the prime ideal $ \pi = \chi \mathcal{O}_F$. It can be identified as 
\begin{equation}
  \mathcal{O}_{\pi} = \{x \in F_{\pi} \mid |x|_{\pi} \leq 1 \}.
\end{equation}

\begin{lemma}
  \label{le:padic_expansion}
  Any element $x \in \mathcal{O}_{\pi}$ can be written as 
  $
      x = x_{0}  + x_{1}  \chi + x_{2} \chi^{2} \dots,
  $
  where $\{x_{i}\}_{i \geq 0} \in \{0,1,2\}$. Furthermore, every element of $F_{\pi}$ can be written as $\chi^{i} x$ for some $i \in \mathbb{Z}_{\leq 0}$ and $x \in \mathcal{O}_{\pi}$.
\end{lemma}

\begin{remark}
 $F_{\pi} = \mathbb{Q}_3[\chi] = \mathbb{Q}_3[\omega]$, $\mathcal{O}_{\pi} = \mathbb{Z}_3[\chi] = \mathbb{Z}_3[\omega]$
\end{remark}
\subsection{Involutions, fixed subrings and bilinear form}

By $\overline{(\cdot)}$, we mean $\omega \mapsto \omega^{-1}$. This defines an automorphism of the rings $F,\mathcal{O}_F, F_{\pi}, \mathcal{O}_{\pi}$ and has the fixed subrings $\mathbb{Q},\mathbb{Z}, \mathbb{Q}_3, \mathbb{Z}_3$ respectively. The involution $\overline{(\cdot)}$ commutes with all the natural inclusions among these rings, wherever applicable. On $\mathcal{O}_{\pi}/\pi \mathcal{O}_{\pi} \simeq \mathbb{F}_3$, $\overline{(\cdot)}$ is the trivial automorphism since $\omega \equiv 1 \pmod{\pi}$ (indeed, $\overline{(\cdot)}$ preserves $\pi \mathcal{O}_{\pi} \subseteq \mathcal{O}_{\pi}$).

The residue field of $\mathbb{Q}_{3}$ is also $\mathbb{F}_{3}$, the same as the residue field of $F_\pi$. 

Wherever the involution $\overline{(\cdot)}$ is defined for a ring $R$, one can define on $R^n$ for any $n$ the bilinear form
\begin{equation}
    \langle\ x,y \ \rangle
 = \sum_{i=1}^n x_i \overline{y_i}.\end{equation}

\subsection{Lattices over local valuation rings}

\begin{definition}
Consider a domain $R$ whose field of fractions is $K$. For a $K$-vector space $V$, 
we define an $R$-lattice in $V$ to be a finitely generated torsion-free $R$-module $\Lambda \subseteq V$ such that $K \Lambda = V$. We denote $\dim V$ to be the rank of the lattice.
\end{definition}

In particular, this means that if we have a $\mathcal{O}_\pi$-lattice $\Lambda$ in $F_{\pi}^3$, then we can conclude that for some $v_1,\dots,v_3 \in F_{\pi}^{3}$, one has
\begin{equation}
  \Lambda = \mathcal{O}_{\pi} v_1 \oplus \mathcal{O}_{\pi} v_{2} \oplus \mathcal{O}_{\pi} v_3.
\end{equation}

This implies that there exists $g \in \M_{3}(F_{\pi})$ such that $\Lambda = g \mathcal{O}_{\pi}^{3}$. This $g$ is not unique, but if  $g \mathcal{O}_{\pi}^{3} = g' \mathcal{O}_{\pi}^{3}$ for some $g'$, then one knows that $g^{-1} g' \in \GL_3(\mathcal{O}_{\pi})$. 

This leads to the following definitions.

\begin{definition}
  Let $\Lambda$ be a lattice over a PID $R$ of rank $n$ and let $K$ be the field of fractions of $R$. We denote $\det(\Lambda) \subseteq K$ to be the $R$-module generated by $\det(g)$, where $\det(g)$ 
  is the determinant of $g \in \M_{n}(K)$ such that $\Lambda = g R^{n}$.
\end{definition}
\begin{remark}
  For an $\mathcal{O}_{\pi}$-lattice $\Lambda$, we know that $ \det(\Lambda) = \pi^{k}$ for some $k \in \mathbb{Z}$.
\end{remark} 
\begin{definition}[Equivalence of $\mathcal{O}_{\pi}$-lattices]
    Two $\mathcal{O}_{\pi}$-lattices $\Lambda_{1},\Lambda_{2}$ are considered $\pi$-equivalent if 
    \begin{equation}
        \Lambda_{1} = \pi^{k}\Lambda_{2} \text{ for some } k \in \mathbb{Z}.
    \end{equation}
\end{definition}

\subsection{Dual lattice}
  \label{se:dual_lattice}

\begin{definition}
	For any $\mathcal{O}_{\pi}$-lattice $\Lambda \subseteq F_{\pi}^{3}$, one defines the dual lattice to be 
	\begin{equation}
		\dual{ \Lambda } = \{x \in F_{\pi}^{3} \mid \langle x,y \rangle \in \mathcal{O}_{\pi} \text{ for each }y \in \Lambda \}.
	\end{equation}
\end{definition}
The dual lattice $\dual{\Lambda}$ is also a $\mathcal{O}_{\pi}$-lattice and upon further dualization gives $\dual{(\dual{\Lambda})}= \Lambda$ .
One then readily checks that if $A \in \M_{3}(F_{\pi})$ is a matrix whose columns are a basis of $\Lambda$, then the conjugate-transpose $( A^{*} )^{-1} \in \M_{3}(F_{\pi})$ contains in its columns a basis of $\dual{\Lambda}$.

\begin{proposition}\label{le:dual_equivalence}
Dualization satisfies the following properties:
\begin{itemize}
\item One has for $\mathcal{O}_{\pi}$-lattices $\Lambda_{1},\Lambda_{2}\subseteq F_{\pi}^{3}$ that
$\Lambda_{1} \subseteq \Lambda_2 \Rightarrow \dual{ \Lambda_{2} }\subseteq \dual{\Lambda_{1}}.$
\item If $\Lambda_{1}$ and $\Lambda_{2}$ are $\pi$-equivalent, then so are $\dual{\Lambda_{1}}$ and $\dual{\Lambda_{2}}$.
\end{itemize}
\end{proposition}

\begin{definition}
A lattice is called self-dual if $\dual{\Lambda}= \Lambda$. 
We say that it is self-dual up to $\pi$-equivalence if $\dual{\Lambda} = \pi^{i} \Lambda$ for some $ i \in \mathbb{Z}$.
\end{definition}
\begin{remark}
The lattice $\Lambda = A \cdot \mathcal{O}_{\pi}^{3}$ is self-dual if $A \in \U_{3}(F_{\pi})$. The converse is not always true.
\end{remark}

For example, $\mathcal{O}_{{\pi}}$ is a self-dual lattice, as is $\mathcal{O}_{\pi}^{3}$.
Here is a useful lemma about duality.
\begin{lemma}
\label{le:dual_scaling}
Suppose $\Lambda$ is a lattice over $\mathcal{O}_{\pi}$ of rank $3$ such that 
$\Lambda$ is $\pi$-equivalent to $\dual{\Lambda}$. That is, 
for some $i \in \mathbb{Z}$, we have 
 $ \dual{\Lambda} = \pi^{i}\Lambda.$
Then, necessarily $i \in 2\mathbb{Z}$ and $\pi^{i/2} \Lambda$ is a self-dual lattice.
\end{lemma}
\begin{proof}
One can check that
that $\det(\dual{\Lambda}) = \overline{ \det(\Lambda) }^{-1}$. Also, we have that $\det(\pi^{i} \Lambda) = \pi^{3i} \det (\Lambda)$.

Hence, we conclude that $\det(\Lambda) \cdot \overline{\det({\Lambda})} = \pi^{-3i}$. But we know that if $\det(\Lambda) = \pi^{k}$ for some $k$, then $\overline{\det(\Lambda)}$ is also $\pi^{k}$ as $\overline{(\cdot)}$ leaves $\pi$ invariant. Therefore, $i$ must be even.
The second statement trivially follows.
\end{proof}
\begin{remark}
  The rank $3$ in Lemma \ref{le:dual_scaling} can be replaced by any other odd number, but fails for even dimensions not equal to $2$.
\end{remark}

\subsection{Cartan decomposition}

Lastly, we recall the following theorem about Cartan decompositions over local fields which we make use of extensively in this paper.
\begin{theorem}[Theorem 3.3 of \cite{prasad2001representation}]
    \label{th:cartan}
Let $g \in \GL_{n}(F)$ be a matrix over a local field $F$ whose ring of integers is $\mathcal{O}$ with uniformizing element $\varpi$. 
Then, one can find $k,k' \in \GL_{n}(\mathcal{O})$ and some diagonal $a \in \GL_{n}(F)$ with entries that are powers of $\varpi$ such that 
$    g= k a k'.$
Furthermore the choice of $a = \diag( \varpi^{ \lambda_{1}}, \dots ,\varpi^{\lambda_{n}})$ is unique if we add the condition
that $\lambda_{1} \geq \lambda_{2} \geq \dots \geq \lambda_{n}$. 
\end{theorem}
\section{Buildings over unitary groups}

\subsection{Lattice chains}

We will consider the lattice chain model of Bruhat-Tits buildings given by Abramenko-Nebe \cite{abramenko2002lattice}. 
\begin{definition}
    \label{de:lattice_chain}
Suppose $R$ is a discrete valuation ring with a uniformizing element $\varpi$ and 
whose field of fractions is $K$. Consider the vector space $K^{n}$ over $K$.

A lattice chain of dimension $n$ over a ring $R$ is defined
	as a nested sequence of lattices with proper inclusions
\begin{equation}
    \cdots \subset \Lambda_{-2} \subset \Lambda_{-1} \subset \Lambda_{0} \subset \Lambda_{1} \subset \Lambda_{2} \subset \cdots
\end{equation}
% I changes to subsets per the reviewer comments, the term proper

	where each of the $\{ \Lambda_{i} \}_{i \in \mathbb{Z}}$ is an $R$-lattice inside $K^{n}$ 
	with the condition that the chain is preserved under the action $\Lambda \mapsto \varpi \Lambda$.
\end{definition}

\begin{remark}
\label{re:lattice_chains_admissible}
In \cite{abramenko2002lattice}, the definition given above is the definition of an ``admissible'' lattice chain and a lattice chain is simply a nested sequence of lattices.
\end{remark}

The infinite chains in Definition \ref{de:lattice_chain} are actually just finitely many equivalence classes. One can prove so in the following lemma.
\begin{lemma}
  \label{le:lattice_chains_are_finite}
  Suppose $\{\Lambda_i\}$ is a lattice chain of dimension $n$ over a discrete valuation ring $R$. Then, up to multiplication by powers of $\varpi$, the chain contains at most $n$ inequivalent lattices.
\end{lemma}
\begin{proof}
    Let's consider $\Lambda_{0}$. We know by Definition \ref{de:lattice_chain} that 
    $\varpi \Lambda_{0}$ is also in the chain. What is the maximum number of lattices $ \varpi \Lambda_{0} \subsetneq \Lambda_{1} \subsetneq \Lambda_{2} \subsetneq \cdots \subsetneq \Lambda_{0}$ that can be between
    $\varpi\Lambda_{0}$ and $ \Lambda_{0}$?

    The answer is $n-1$. Indeed, notice that $\Lambda_{0}/ \varpi\Lambda_{0} \simeq k^{n}$ where $k = R/\varpi R$. We know that $k$ is a field. One can show that the intermediate lattices $\Lambda_{1}, \Lambda_{2}, \dots$ form a flag of subspaces of $k^{n}$. Hence, they are finitely many and at most $n-1$.
\end{proof}

\begin{definition}
    We define 
    the rank of a lattice chain $\{\Lambda_{i}\}$
    to be 
    the number of $\varpi$-equivalence classes inside a lattice chain.
\end{definition}

We say that admissible lattice chain $\{L_i\}$ contains $\{ L_i'\}$ if the latter chain appears as a subsequence in the former chain. If $\{L_{i}\}$ has rank $n$, the $\{L_{i}'\}$ has rank at most $n$.

In this paper, we will only consider lattice chains relevant for qutrit synthesis. In our setting, $\varpi = \pi$, $K = F_\pi$, $R= \mathcal{O}_\pi$ and $n=3$. A lattice chain $\dots \subset \Lambda_{-2} \subset \Lambda_{-1} \subset \Lambda_0 \subset \Lambda_1 \subset \Lambda_2 \subset \dots$ will consist of $\mathcal{O}_\pi$ lattices over $F_\pi^3$.

\subsection{Self-dual lattice chains}

Observe that when we have a nested chain of lattices we can dualize each lattice in the chain and obtain
\begin{equation}
    \cdots \subset \Lambda_{-2} \subset \Lambda_{-1} \subset \Lambda_{0} \subset \Lambda_{1} \subset  \Lambda_{2} \subset \cdots \quad \longmapsto \quad
    	\cdots \supset \dual{\Lambda}_{-2} \supset \dual{\Lambda}_{-1} \supset \dual{\Lambda}_{0} \supset \dual{\Lambda}_{1} \supset  \dual{\Lambda}_{2} \supset \cdots
\end{equation}
The right hand side is called the dual lattice chain.
Our interest is in self-dual lattice chains. 
% This is defined in the following way.
\begin{definition}
A lattice chain over $\mathcal{O}_{\pi}$ is called a self-dual lattice chain if it is the same chain when we dualize each lattice in the chain.
\end{definition}

\begin{lemma}
    \label{selfdual_oradjacent}
    For a self-dual lattice chain $\{\Lambda_i\}_{i \in \mathbb{Z}}$, at least one of the two statements are true.
\begin{enumerate}
\tightlist
    \item One has $\Lambda_j = \dual{\Lambda_j}$ for some $j \in \mathbb{Z}$, or
    \item One has $\Lambda_j = \dual{\Lambda_{j+1}}$ for some $j \in \mathbb{Z}$.
\end{enumerate}
\end{lemma}
\begin{proof}

    A lattice chain is equal to its dual if and only if there exists a bijection $d: \mathbb{Z} \rightarrow \mathbb{Z}$ such that $\Lambda_{i} = \dual{\Lambda_{d(i)}}$ that respects the condition $d(i-1) = d(i)+1$. This map enforces that the lattices in the dual chain are the same as in the original chain, but in a reversed order. 
    \begin{enumerate}
    \tightlist
        \item Let $2j \in \mathbb{Z}$ satisfy $d(2j) = 0$.  $d(j) = d(2j - j)= j$ implies that $\Lambda_j = \dual{\Lambda_j}$ 
        \item Let $2j +1 \in \mathbb{Z}$ satisfy $d(2j+1) = 0$.  $d(j) = d(2j +1 - (j+1))= j+1$ implies that $\Lambda_j = \dual{\Lambda_{j+1}}$ 
    \end{enumerate} 
\end{proof}

\subsection{Simplicial complex of lattice chains}

Now we are ready to define the building on which our finitely generated gate set will act. 
\begin{definition}
We define $\mathcal{B}$ to be the set of all self-dual lattice chains of rank $3$ over $\mathcal{O}_{\pi}$. This is the building associated to the unitary group.
\end{definition}

\subsubsection{Description of 0-simplices and 1-simplices}

As a set, $\mathcal{B}$ is just a collection of chains. But it actually has the structure of a simplicial complex. The $0$-simplices of $\mathcal{B}$ are the minimal lattice chains. The $1$-simplices are lattice chains that strictly contain only $0$-simplices. The $2$-simplices are lattice chains that only contain $1$-simplices and $0$-simplices, and so on.

\begin{remark}
    The set of all lattice chains (not necessarily self-dual) form a 2-building for the group $\PGL_3$. In the $\PGL_3$-building, lattice chains of rank $k$ form the $(k-1)$-simplices for $k=1,2,3$.
    The building associated to $\U_{3}$ is a sub-building of this building.
    We will not discuss the building of $\PGL_3$ for simplicity of exposition.
\end{remark}

We shall prove the following properties about $\mathcal{B}$ (c.f. Appendix \ref{ap:simplex_classification}).
\begin{proposition}
    \label{pr:simplex_classification}
For the $\mathcal{B}$ described above, we have the following classification of simplices.
  \begin{enumerate}
  \tightlist
    \item All the $0$-simplices  of $\mathcal{B}$ are one of the following two types.
      \begin{itemize}
\tightlist
          \item 
Lattice chains of the form $\{ \pi^{i}\Lambda\}_{i \in \mathbb{Z}}$ for some self-dual lattice $\Lambda \subseteq F_{\pi}^{3}$  arranged as 
\begin{equation}
    \dots \subseteq \pi^{2} \Lambda \subseteq \pi \Lambda \subseteq \Lambda \subseteq \pi^{-1} \Lambda \subseteq \dots
\end{equation}
    \item 
	Lattice chains of the form $\{ \pi^i\Lambda, \pi^i \dual{\Lambda}\}_{i \in \mathbb{Z}}$ 
	 such that $\dual{\Lambda} \subsetneq\Lambda$ and
	 \begin{equation}
	       \dots \subseteq \pi \Lambda \subseteq \dual{\Lambda}  \subseteq \Lambda \subseteq \pi^{-1} \dual{\Lambda}  \subseteq \dots
	 \end{equation}
      \end{itemize}
  \item All the 1-simplices of $\mathcal{B}$ are of the form $\{ \pi^i\Lambda_1, \pi^i \Lambda_0, \pi^i \dual{\Lambda_{1}} \}_{i \in \mathbb{Z}}$  for some lattices $\Lambda_0,\Lambda_1\subseteq F_{{\pi}}^3$ 
      such that $\Lambda_0$ is self dual and $\Lambda_1^{\dual{}} \subsetneq \Lambda_1$ arranged as
      \begin{equation}
              \dots \subseteq \pi \Lambda_0 \subseteq \pi\Lambda_{1}  \subseteq \dual{\Lambda_{1}} \subseteq \Lambda_0 \subseteq \Lambda_{1} \subseteq \pi^{-1} \dual{\Lambda_{1}} \subseteq \dots
      \end{equation}
      
      \item There are no 2-simplices in $\mathcal{B}$.
      \end{enumerate}
\end{proposition}

We will henceforth think of $0$-cells as vertices and $1$-cells as edges. Therefore, $\mathcal{B}$ has the structure of a graph.
We will call $0$-cells of the form $\{ \pi^{i}\Lambda\}_{ i \in \mathbb{Z}}$ for a self-dual lattice $\Lambda$ as ``pure'' vertices. We will call the $0$-cells that are not ``pure'' as ``alternating'' cells.
We will denote pure cells to be $P_{\mathcal{B}}$ and alternating as $A_{\mathcal{B}}$. One then has that vertex set of $\mathcal{B}$ decomposes as 

\[\text{ Vertices of }\mathcal{B} = P_{\mathcal{B}} \sqcup A_\mathcal{B}.\]
Denote $e_{\mathcal{B}}$ to be the distinguished $0$-cell given by $\{\pi^{i} \mathcal{O}_{\pi}^{3}\}_{i \in \mathbb{Z}}$ which we will often call the ``origin'' (also written $v_0$ or $e_0$ below).

\begin{remark}
    $P_\mathcal{B}$ and $A_\mathcal{B}$ split $\mathcal{B}$ into a bipartite graph.
\end{remark}

One can summarize the structure of $\mathcal{B}$ according to the following proposition (c.f. Appendix \ref{se:forest_proof})
\begin{proposition} \label{pr:forest}
    The $0$-simplices and $1$-simplices of $\mathcal{B}$ form a forest. That is, $\mathcal{B}$ has no loops.
\end{proposition}
\subsubsection{Some subspaces in finite field vector spaces}

We are interested in establishing some general facts about quadratic forms on vector spaces over $\mathbb{F}_{3}$.

First we consider the case of $\mathbb{F}_{3}^{3}$.
  Consider $\mathbb{F}_{3} ^{3}$ with the standard inner product $\langle \ , \ \rangle : \mathbb{F}_{3}^{3} \times \mathbb{F}_{3}^{3} \to \mathbb{F}_{3}$ given by: 
  \begin{equation}
      \label{eq:defi_of_angle3}
      \langle x^{(1)},x^{(2)} \rangle = x^{(1)}_{1} x^{(2)}_{1} + 
 x^{(1)}_{2} x^{(2)}_{2} + 
 x^{(1)}_{3} x^{(2)}_{3}  
  \end{equation}
  where $ x^{(i)} = (x^{(i)}_{1},
 x^{(i)}_{2}, x^{(i)}_{3}) \in \mathbb{F}_{3}^{3} \text{ for } i = 1,2.$ 

 For any bilinear form $\langle \ , \ \rangle:\mathbb{F}_{3}^{3} \times \mathbb{F}_{3}^{3} \rightarrow  \mathbb{F}_{3}$, one defines the dual space of a vector space
 $V \subseteq \mathbb{F}_{3}^{3}$ with respect to $\langle\ , \ \rangle$ to be 
 \begin{equation}
     V^{\perp} = 
     \{v \in  \mathbb{F}_{3}^{3} \mid  \langle v,w \rangle = 0 \text{ for all }w \in V \}.
 \end{equation}

\begin{lemma}
  \label{le:mod3_codes}
  Consider $\mathbb{F}_{3}^{3}$ and let $A \in \M_{3}(\mathbb{F}_{3})$ be a symmetric invertible matrix. 
  Consider the bilinear form $\langle\ , \ \rangle_{A}$ given by 
  \begin{equation}
      \langle x,y \rangle _{A} =  \langle A x, y \rangle.
  \end{equation}
  Then, there are exactly $4$ one-dimensional subspaces $V \subseteq \mathbb{F}_{3}^{3}$ such that $V \subseteq V^{\perp}$, where $V^{\perp}$ is the dual of $V$ with respect to $\langle \ , \ \rangle_{A}$.
\end{lemma}
\begin{proof}
    We know by \cite{Lidl_1996} that in a finite field of odd characteristic, every non-degenerate quadratic form can be diagonalized.
    Hence, for $x = (x_{1},x_{2},x_{3}) \in \mathbb{F}_{3}^{3}$, one gets that there exists some $B \in \GL_{3}(\mathbb{F}_{3})$ such that 
    \begin{equation}
	q(B x) = \langle Bx, ABx \rangle = a x_{1}^{2} + bx_{2}^{2} + c x_{3}^{2}
    \end{equation}
    for some $a,b,c \in \{\pm 1 \}$. 
We want to count the number of $x \in \mathbb{F}_{3}^{3}\setminus \{0\}$ such that $q(Bx) = 0$.
    For this purpose one sees that without loss of generality, we can assume that 
    \begin{equation}
        q(Bx) = x_{1}^{2} + x_{2}^{2} + x_{3}^{2} \text{ or } x_{1}^{2} + x_{2}^{2} - x_{3}^{2}.
    \end{equation}
    The exact diagonalized form of $q(Bx)$ depends on whether $\det A$ is a quadratic residue in $\mathbb{F}_{3}$.
    In both the cases, we have exactly $8$ solutions and they must give rise $4$ subspaces $V \subseteq \mathbb{F}_{3}^{3}$ that are one-dimensional. For each $V$, it is clear that $V \subseteq V^{\perp}$ and one sees that these are all.
\end{proof}

We also need to consider the situation of anti-symmetric linear forms.
\begin{lemma}
\label{antisym}
Consider $\mathbb{F}_{3}^{3}$ and an anti-symmetric matrix $A \in \M_{3}(\mathbb{F}_{3})$ given by the following for some $(a,b) \in \mathbb{F}_{3}^{2} \setminus \{(0,0)\}$:
\begin{equation}
    A 
    =  \left(
    \begin{smallmatrix}
	\ &  a & b \\ 
	-a & & \\ 
	-b & & 
    \end{smallmatrix}
    \right) .
\end{equation}
Consider the bilinear form $\langle\ , \ \rangle_{A}$ given by 
\begin{equation}
    \langle x,y\rangle_{A}  = \langle A x, y \rangle.
\end{equation}

Then, the vector $v=(0,b,-a)^{T} \in \mathbb{F}_{3}^{3}$ is the unique vector up to scaling satisfying $\langle v ,w \rangle_{A} = 0$ for all $w \in \mathbb{F}_{3}^{3}$. Furthermore, there are exactly four subspaces $V \subseteq \mathbb{F}_{3}^{3}$ of $\mathbb{F}_{3}$-dimension 2 such that $\mathbb{F}_{3} \cdot v \subseteq V$ and $V^{\perp} = \{ w \in \mathbb{F}_{3}^{3} \mid \langle u,w \rangle_{A} = 0 \text{ for all }u \in V\} = V$.
\end{lemma}
\begin{proof}
Proof is in Appendix \ref{proofappendixc}.
\end{proof}

\subsection{Understanding the tree using finite fields}

\subsubsection{Pure vertices}

For any $\mathcal{O}_{\pi}$-lattice $\Lambda \subseteq F_{\pi}^{3}$, one has that $\Lambda / \pi \Lambda $ is a finite vector space over $\mathcal{O}_{\pi}/\pi \mathcal{O}_{\pi} \simeq \mathbb{F}_{3}$ of dimension $3$.
Let us consider a pure $0$-simplex $\{\pi^{i} \Lambda\}_{i \in \mathbb{Z}}$ for $\Lambda \subseteq F_{\pi}^{3}$ a self-dual lattice. 
One then gets the following ``$\langle \ , \ \rangle$ modulo $\pi$'' map:
\begin{equation}
    \langle \ , \ \rangle + \pi \mathcal{O}_{\pi} : \frac{\Lambda}{\pi \Lambda} \times \frac{\Lambda}{ \pi \Lambda } \to \frac{\mathcal{O}_{\pi}}{ \pi \mathcal{O}_{\pi}}.
    \label{eq:defi_of_angle}
\end{equation}
Indeed, this is well-defined since self-duality of $\Lambda$ implies $\langle \Lambda , \Lambda \rangle \subseteq \mathcal{O}_{\pi}$ whereas $\langle \Lambda , \pi \Lambda \rangle \subseteq \pi \mathcal{O}_{\pi}$.

\begin{lemma}
  \label{le:nondegen}
  Consider $\Lambda \subseteq F_{\pi}^{3}$ to be a self-dual lattice. Then the bilinear form $\langle \ , \ \rangle + \pi \mathcal{O}_{\pi}$ given in Equation \ref{eq:defi_of_angle} defined on the finite $\mathbb{F}_{3}$-vector space $\Lambda / \pi \Lambda$ is non-degenerate and symmetric.
\end{lemma}
\begin{proof}
    The symmetry is trivial to check since $x + \pi \mathcal{O}_{\pi} = \overline{x} + \pi\mathcal{O}_{\pi}$ for any $x \in \mathcal{O}$.
        What we want to then check is that for any $v \in \Lambda$, $\langle v ,w \rangle \in \pi \mathcal{O}_{\pi}$ for all $w \in \Lambda$ then $v \in \pi \Lambda$. But this is clear since we must then have $\langle \chi^{-1}  v , w \rangle \in \mathcal{O}_{\pi}$ for all $w \in \Lambda$ forcing that $\chi^{-1} v \in \dual{\Lambda} = \Lambda$.
\end{proof}

This makes us conclude the following fact about the geometry of the Bruhat-Tits tree $\mathcal{B}$.

\begin{proposition}
    Each pure vertex $x \in \mathcal{P}_{\mathcal{B}}$ is connected to $4$ alternating vertices in $A_{\mathcal{B}}$. 
\end{proposition}
\begin{proof}
    Let $x = \{\pi^{i} \Lambda \}_{i \in \mathbb{Z}}$ such that $\Lambda \subseteq {F}_{\pi}^{3}$ is a self-dual lattice. 
    The goal is to find all possible lattices $\Lambda_{1} \subseteq F_{\pi}^{3}$ such that 
\begin{equation}
    \dots \subset \pi \Lambda \subset \pi \Lambda_{1} \subset \dual{\Lambda_{1}} \subset \Lambda \subset \Lambda_{1} \subset \pi^{-1} \dual{\Lambda_{1}} \subset \pi^{-1} \Lambda \subset \dots
\end{equation}

Fix a vector space isomorphism $\frac{\Lambda}{\pi \Lambda} \simeq \mathbb{F}_{3}^{3}$. 
Then, we know that the bilinear form $\langle \ ,\ \rangle + \pi \mathcal{O}_{\pi}$
defines a quadratic form on $\mathbb{F}_{3}^{3}$. 
This form is non-degenerate and symmetric so we can invoke Lemma \ref{le:mod3_codes}.

One can check that each chain $\pi \Lambda \subset \pi \Lambda_{1} \subset \dual{\Lambda_{1}} \subset \Lambda$ corresponds bijectively to a chain of subspaces $\{0\} \subset V \subset V^{\perp} \subset \mathbb{F}_{3}^3$ where $V$ has $\mathbb{F}_{3}$-dimension equal to 1 and $V^{\perp}$ is the dual of $V$ with respect to $\langle\ , \ \rangle + \pi \mathcal{O}_{\pi}$. 
The number of such $V$ is exactly $4$ as shown in Lemma \ref{le:mod3_codes}.

\end{proof}

\begin{figure}
    \centering
    \begin{tikzpicture}[scale=0.6,
        dot/.style={circle, fill=black, inner sep=3pt},
        lababove/.style={font=\footnotesize, align=center, above=18pt},
        lableft/.style={font=\footnotesize, align=right, left=7pt},
        labright/.style={font=\footnotesize, align=left, right=7pt}
    ]

    \node[dot] (pure) at (0,0) {};
    \node[lababove] at (pure.north) {Pure vertex\\ $\Lambda/\pi\Lambda \cong \mathbb{F}_3^3$};

    \node[dot] (alt1) at ( 2,  2) {};
    \node[dot] (alt2) at ( 2, -2) {};
    \node[dot] (alt3) at (-2, -2) {};
    \node[dot] (alt4) at (-2,  2) {};

    \node[labright] at (alt1.east) {Alternating vertex\\ $V_1 \subset V_1^\perp$};
    \node[labright] at (alt2.east) {Alternating vertex\\ $V_2 \subset V_2^\perp$};

    \node[lableft] at (alt3.west) {Alternating vertex\\ $V_3 \subset V_3^\perp$};
    \node[lableft] at (alt4.west) {Alternating vertex\\ $V_4 \subset V_4^\perp$};

    \draw[line width=1.5pt] (pure) -- (alt1);
    \draw[line width=1.5pt] (pure) -- (alt2);
    \draw[line width=1.5pt] (pure) -- (alt3);
    \draw[line width=1.5pt] (pure) -- (alt4);

    \end{tikzpicture}
    \caption{A pure vertex connected to four alternating vertices via the four isotropic lines $V \subset V^\perp$ in $\mathbb{F}_3^3$.}
    \label{fig:purevertex_x}
\end{figure}

    \subsubsection{Alternating vertices}

    We will now start with an alternating $0$-simplex $v = \{ \pi^{i}\Lambda, \pi^{i}\dual{\Lambda} \}_{i \in \mathbb{Z}}$ such that 
    \begin{equation}
	\label{eq:chain_alternating}
	\cdots   \subset \pi \Lambda \subset  \dual{\Lambda} \subset  \Lambda \subset   \pi^{-1} \dual{\Lambda} \subset \cdots 
    \end{equation}
    The goal is now to find how many self-dual lattices $\Lambda_{1} \subseteq F_{\pi}^{3}$ one can find, if any, such that 
    \begin{equation}
	\cdots \subset \pi \Lambda_{1}   \subset \pi \Lambda \subset  \dual{\Lambda} \subset \Lambda_{1}  \subset  \Lambda \subset   \pi^{-1} \dual{\Lambda} \subset \pi^{-1} \Lambda_{1} \subset\cdots 
	\label{eq:chain_sandiwich}
    \end{equation}

    The image of $\dual{\Lambda}$ in the $\mathbb{F}_{3}$-vector space $\Lambda/ \pi \Lambda$ must be a $1$-dimensional vector space. 
    One can indeed check this by observing that if $ \Lambda = g \mathcal{O}_{\pi}^3$ then $[\Lambda:\dual{\Lambda}] = 3^{-2v_{\pi}(\det g)}$ is an even power of $3$ (as $N(\det g)=\det(g)\overline{\det(g)}$ has even $\pi$-valuation) dividing $[\Lambda:\pi\Lambda]=27$; since $\Lambda$ is not self-dual it is forced to be $9$.

    Hence, we can consider $\Lambda / \pi {\Lambda} \simeq \mathbb{F}_{3}^{3}$. 
    The relevant bilinear form on this $\mathbb{F}_{3}$-space is 
    \begin{equation}
	\label{eq:defi_of_angle4}
	\langle \ , \ \rangle + \mathcal{O}_{\pi}
	: \frac{\Lambda}{\pi {\Lambda} } \times
	\frac{\Lambda}{\pi {\Lambda}} 
	\rightarrow  \frac{\pi^{-1} \mathcal{O}_{\pi}}{ \mathcal{O}_{\pi}} \simeq \mathbb{F}_{3}. 
    \end{equation}
    One can check that this is indeed well-defined. 
    First observe that $\langle \Lambda , \Lambda \rangle \subseteq \pi^{-1} \langle \Lambda , \dual{\Lambda} \rangle \subseteq \pi^{-1} \mathcal{O}_{\pi}$ makes sense whereas $\langle \dual{\Lambda}, \Lambda \rangle \subseteq \mathcal{O}_{\pi}$ is in the kernel.

    To have an $\mathcal{O}_{\pi}$-valued form on $\Lambda/\pi\Lambda$, one can consider the form 
	$\chi \langle \ , \ \rangle + \pi \mathcal{O}_{\pi} $ instead. The next two lemmas are about this bilinear form.

	\begin{lemma}
	  \label{le:techincal_lemma}
	  Let $x,y \in F_{\pi}$ be such that $\overline{x} x + \overline{y} y \in \mathcal{O}_{\pi}$. Then, $x,y \in \mathcal{O}_{\pi}$.
	\end{lemma}
    \begin{proof}
    The proof can be found in Appendix \ref{proofoflemma3.16}
    \end{proof}
	    % If $v_{\pi}(x) < v_{\pi}(y)$ then we clearly have that $ 0 \leq v_{\pi}( \overline{x} x + \overline{y} y )  = 2 v_{\pi}(x)$ and we are done.

	    % Otherwise, let's assume that $x = \chi^{n}x', y=\chi^{n}y'$ for some $n < 0$ such that $x',y' \in \mathcal{O}_{\pi}^{\times}$.
	    % Then, we have $x',y' \in \pm 1 + \pi \mathcal{O}_{\pi}$ implying that $\overline{x'} x' + \overline{y'} y' \in  2 + \pi \mathcal{O}_{\pi}$. Therefore, we must have that $x \overline{x} + y \overline{y} \in \chi^{n} \overline{\chi}^{n}(2 + \pi \mathcal{O}_{\pi} )$ forcing $n$ to be non-negative. This is impossible. So $x,y \in \mathcal{O}_{\pi}$.

\begin{lemma}
  \label{le:nondegen2}
  Consider the form $\langle \ , \ \rangle + \mathcal{O}_{\pi}$ defined in 
  Equation \ref{eq:defi_of_angle4}. Then, after identification of $\Lambda / \pi \Lambda \simeq \mathbb{F}_{3}^{3}$, the map $\chi \langle \ , \ \rangle + \pi \mathcal{O}_{\pi}$ is anti-symmetric and is equivalent to the bilinear map $\mathbb{F}_{3}^{3} \times \mathbb{F}_{3}^{3} \rightarrow  \mathbb{F}_{3}$ considered in Lemma \ref{antisym}.
\end{lemma}
\begin{proof}
The proof can be found in Appendix \ref{proofofprop5.19}.
\end{proof}
Furthermore we have the following proposition:
\begin{proposition}
    \label{pr:alternatings_has_pure_neighbours}
    Each alternating vertex $v \in \mathcal{A}_{\mathcal{B}}$ is connected to $4$ pure vertices in $\mathcal{P}_{\mathcal{B}}$.
\end{proposition}
\begin{proof}

Let $v = \{ \pi^{i}\Lambda, \pi^{i} \dual{\Lambda}\}_{i \in \mathbb{Z}}$ where $\Lambda \subseteq F_{\pi}^{3}$ is a lattice satisfying Equation \ref{eq:chain_alternating}. We know from Lemma \ref{le:nondegen2} that the bilinear form $\chi \langle \ , \ \rangle + \pi \mathcal{O}_{\pi}$ on $\Lambda/\pi \Lambda \simeq \mathbb{F}_{3}^{3}$ must be isomorphic to the one in Lemma \ref{antisym}, up to some linear transformations on $\mathbb{F}_{3}^3$. The subspace $\mathbb{F}_{3} \cdot v $ for the vector $v= (0,b,-a)^{T} \in \mathbb{F}_{3}^{3}$ must correspond to the image of $\dual{\Lambda}$ in $\Lambda/\pi \Lambda$. The four self-dual subspaces $V$ in Lemma \ref{antisym} therefore lift to four self-dual lattices $\Lambda_{1}$ satisfying Equation \ref{eq:chain_sandiwich}.

\end{proof}

\subsubsection{The tree is connected}

We will now show that the tree $\mathcal{B}$ is a connected tree. First, observe the easy lemma.
\begin{lemma}
  \label{le:gramm_matrix}
  Let $ \Lambda =  v_{1}\mathcal{O}_{\pi}+ v_{2} \mathcal{O}_{\pi} + v_{3}\mathcal{O}_{\pi} \subseteq F_{\pi}^{3}$ be a lattice. Then, $\Lambda$ is self-dual if and only if the Gram matrix $[ \langle v_{i}, v_{j} \rangle]_{1 \leq i ,j \leq 3} \in \GL_{3}(\mathcal{O}_{\pi})$.
\end{lemma}
\begin{proof}
If $A \in \GL_3(F_\pi)$ is a matrix with columns being  $v_1,v_2,v_3$, we know by the discussion in Section \ref{se:dual_lattice}
that the matrix $(A^*)^{-1}$ contains a basis of 
$\dual{\Lambda} = \Lambda$. 
Hence, for some $\gamma \in \GL_3(\mathcal{O}_\pi)$, 
one has $(A^*)^{-1} = A \gamma$.  
This implies that $A^* A \in \GL_3(\mathcal{O}_\pi)$ and it can be checked that $A^* A$ is exactly the Gram matrix.
Conversely, if $A^* A \in \GL_3(\mathcal{O}_\pi)$ then $(A^*)^{-1} = A (A^* A)^{-1} \in A\,\GL_3(\mathcal{O}_\pi)$, so $\dual{\Lambda} = (A^*)^{-1}\mathcal{O}_\pi^3 = A \mathcal{O}_\pi^3 = \Lambda$.
\end{proof}

We will show the following key proposition towards connectedness. 
\begin{proposition}
\label{pr:chain}
    Let $g,h \in \GL_{3}(F_{\pi})$ be such that $\Lambda_{g} = g \mathcal{O}_{\pi}^{3}$ and $\Lambda_{h} =h \mathcal{O}_{\pi}^{3}$ are self-dual lattices. Furthermore, we assume that 
    \begin{equation}
	\dots \subseteq \pi^{n}\Lambda_{g} \subseteq \Lambda_{h} \subseteq \pi^{-n}\Lambda_{g} \subseteq  \pi^{-2n} \Lambda_{h} \subseteq \dots,
	\label{eq:non-lattice-chain}
    \end{equation}
    where $n \in \mathbb{Z}_{\geq 0}$ is the least such $n$ that makes Equation \ref{eq:non-lattice-chain} hold.
    Then, there exists a sequence of self-dual lattices $\Lambda_{1}, \Lambda_{2} ,\dots, \Lambda_{n-1} \subseteq F_{\pi}^{3}$ such that 
    \begin{align}
	\label{eq:non-lattice-chain2}
	\dots \subseteq \pi \Lambda_{1} \subseteq &  \Lambda_{g} \subseteq  \pi^{-1} \Lambda_{1} \subseteq \pi^{-2}\Lambda_{g} \subseteq \dots \\
	 \dots \subseteq \pi \Lambda_{2} \subseteq &  \Lambda_{1} \subseteq  \pi^{-1} \Lambda_{2} \subseteq \pi^{-2}\Lambda_{1} \subseteq \dots \\
						   & \vdots \\
	 \dots \subseteq \pi \Lambda_{n-1} \subseteq &  \Lambda_{h} \subseteq  \pi^{-1} \Lambda_{n-1} \subseteq \pi^{-2}\Lambda_{h} \subseteq \dots \\
    \end{align}
\end{proposition}
\begin{proof}
The proof can be found in Appendix \ref{proofofprop5.22}.
\end{proof}
\begin{corollary}
The tree $\mathcal{B}$ is connected.
\end{corollary}
\begin{proof}
    It is true that for any self-dual lattices $\Lambda_{g}$, $\Lambda_{h} \in \mathcal{P}_{\mathcal{B}}$ one can find some minimal $n \in \mathbb{Z}_{\geq 0}$ such that 
   Equation \ref{eq:non-lattice-chain} is true. In Lemma \ref{le:hecke_nghbs}, we will see that if $n=1$, one can connect such pure vertices by an alternating pair. If $n>1$, we can reduce to the case of $n=1$ by Proposition \ref{pr:chain}.

   To connect $x,y \in \mathcal{A}_{\mathcal{B}}$, one can use Proposition \ref{pr:alternatings_has_pure_neighbours} to connect them to pure neighbors and assume $x,y \in \mathcal{P}_{\mathcal{B}}$ without loss of generality.
\end{proof}
                  We have plotted in Figure \ref{fig:Btbuilding2}(a) the appearance of the tree with the structure described above.

\section{Clifford+R action on the Bruhat-Tits building}

Observe that $\U_{3}(F_{\pi})$ has a natural action on $\mathcal{B}$. 
The stabilizer of the origin $e_{\mathcal{B}}$ is $\U_{3}(\mathcal{O}_{\pi})$. 

\begin{figure}
    \centering
    \begin{subfigure}
    {0.48\textwidth}
    \hspace{2cm}
    \label{fig:Bt building}
\begin{tikzpicture}[scale=0.12, every node/.style={circle, draw, inner sep=1pt}]
  \node[fill=red] (n0) at (0.000,0.000) {};
  \node[fill=blue] (n1) at (4.243,-4.243) {};
  \node[fill=blue] (n2) at (4.243,4.243) {};
  \node[fill=blue] (n3) at (-4.243,4.243) {};
  \node[fill=blue] (n4) at (-4.243,-4.243) {};
  \node[fill=red] (n5) at (3.106,-11.591) {};
  \node[fill=red] (n6) at (8.485,-8.485) {};
  \node[fill=red] (n7) at (11.591,-3.106) {};
  \node[fill=red] (n8) at (11.591,3.106) {};
  \node[fill=red] (n9) at (8.485,8.485) {};
  \node[fill=red] (n10) at (3.106,11.591) {};
  \node[fill=red] (n11) at (-3.106,11.591) {};
  \node[fill=red] (n12) at (-8.485,8.485) {};
  \node[fill=red] (n13) at (-11.591,3.106) {};
  \node[fill=red] (n14) at (-11.591,-3.106) {};
  \node[fill=red] (n15) at (-8.485,-8.485) {};
  \node[fill=red] (n16) at (-3.106,-11.591) {};
  \node[fill=blue] (n17) at (1.656,-18.928) {};
  \node[fill=blue] (n18) at (4.918,-18.353) {};
  \node[fill=blue] (n19) at (8.030,-17.220) {};
  \node[fill=blue] (n20) at (10.898,-15.564) {};
  \node[fill=blue] (n21) at (13.435,-13.435) {};
  \node[fill=blue] (n22) at (15.564,-10.898) {};
  \node[fill=blue] (n23) at (17.220,-8.030) {};
  \node[fill=blue] (n24) at (18.353,-4.918) {};
  \node[fill=blue] (n25) at (18.928,-1.656) {};
  \node[fill=blue] (n26) at (18.928,1.656) {};
  \node[fill=blue] (n27) at (18.353,4.918) {};
  \node[fill=blue] (n28) at (17.220,8.030) {};
  \node[fill=blue] (n29) at (15.564,10.898) {};
  \node[fill=blue] (n30) at (13.435,13.435) {};
  \node[fill=blue] (n31) at (10.898,15.564) {};
  \node[fill=blue] (n32) at (8.030,17.220) {};
  \node[fill=blue] (n33) at (4.918,18.353) {};
  \node[fill=blue] (n34) at (1.656,18.928) {};
  \node[fill=blue] (n35) at (-1.656,18.928) {};
  \node[fill=blue] (n36) at (-4.918,18.353) {};
  \node[fill=blue] (n37) at (-8.030,17.220) {};
  \node[fill=blue] (n38) at (-10.898,15.564) {};
  \node[fill=blue] (n39) at (-13.435,13.435) {};
  \node[fill=blue] (n40) at (-15.564,10.898) {};
  \node[fill=blue] (n41) at (-17.220,8.030) {};
  \node[fill=blue] (n42) at (-18.353,4.918) {};
  \node[fill=blue] (n43) at (-18.928,1.656) {};
  \node[fill=blue] (n44) at (-18.928,-1.656) {};
  \node[fill=blue] (n45) at (-18.353,-4.918) {};
  \node[fill=blue] (n46) at (-17.220,-8.030) {};
  \node[fill=blue] (n47) at (-15.564,-10.898) {};
  \node[fill=blue] (n48) at (-13.435,-13.435) {};
  \node[fill=blue] (n49) at (-10.898,-15.564) {};
  \node[fill=blue] (n50) at (-8.030,-17.220) {};
  \node[fill=blue] (n51) at (-4.918,-18.353) {};
  \node[fill=blue] (n52) at (-1.656,-18.928) {};
  \draw (n0) -- (n1);
  \draw (n0) -- (n2);
  \draw (n0) -- (n3);
  \draw (n0) -- (n4);
  \draw (n1) -- (n5);
  \draw (n1) -- (n6);
  \draw (n1) -- (n7);
  \draw (n2) -- (n8);
  \draw (n2) -- (n9);
  \draw (n2) -- (n10);
  \draw (n3) -- (n11);
  \draw (n3) -- (n12);
  \draw (n3) -- (n13);
  \draw (n4) -- (n14);
  \draw (n4) -- (n15);
  \draw (n4) -- (n16);
  \draw (n5) -- (n17);
  \draw (n5) -- (n18);
  \draw (n5) -- (n19);
  \draw (n6) -- (n20);
  \draw (n6) -- (n21);
  \draw (n6) -- (n22);
  \draw (n7) -- (n23);
  \draw (n7) -- (n24);
  \draw (n7) -- (n25);
  \draw (n8) -- (n26);
  \draw (n8) -- (n27);
  \draw (n8) -- (n28);
  \draw (n9) -- (n29);
  \draw (n9) -- (n30);
  \draw (n9) -- (n31);
  \draw (n10) -- (n32);
  \draw (n10) -- (n33);
  \draw (n10) -- (n34);
  \draw (n11) -- (n35);
  \draw (n11) -- (n36);
  \draw (n11) -- (n37);
  \draw (n12) -- (n38);
  \draw (n12) -- (n39);
  \draw (n12) -- (n40);
  \draw (n13) -- (n41);
  \draw (n13) -- (n42);
  \draw (n13) -- (n43);
  \draw (n14) -- (n44);
  \draw (n14) -- (n45);
  \draw (n14) -- (n46);
  \draw (n15) -- (n47);
  \draw (n15) -- (n48);
  \draw (n15) -- (n49);
  \draw (n16) -- (n50);
  \draw (n16) -- (n51);
  \draw (n16) -- (n52);
\end{tikzpicture}
\caption{Vertices up to distance 3 from the center of $\mathcal{B}$. Note that the pure vertices are denoted by red and the alternating ones by blue.}
\end{subfigure}
\hspace{0.5cm}
\begin{subfigure}{0.40\textwidth}
\centering
    \centering
\begin{tikzpicture}[scale=0.14, every node/.style={circle, draw, inner sep=1.5pt}]
  \node[fill=red] (n0) at (0.000,0.000) {$e_0$};
  \node[fill=blue] (n1) at (4.243,-4.243) {};
  \node[fill=blue] (n2) at (4.243,4.243) {};
  \node[fill=blue] (n3) at (-4.243,4.243) {};
  \node[fill=blue] (n4) at (-4.243,-4.243) {};
  \node[fill=green] (n5) at (3.106,-11.591) {$e$};
  \node[fill=red] (n6) at (8.485,-8.485) {};
  \node[fill=red] (n7) at (11.591,-3.106) {};
  \node[fill=pink] (n8) at (11.591,3.106) {$e_1$};
  \node[fill=red] (n9) at (8.485,8.485) {};
  \node[fill=red] (n10) at (3.106,11.591) {};
  \node[fill=red] (n11) at (-3.106,11.591) {};
  \node[fill=red] (n12) at (-8.485,8.485) {};
  \node[fill=red] (n13) at (-11.591,3.106) {};
  \node[fill=red] (n14) at (-11.591,-3.106) {};
  \node[fill=red] (n15) at (-8.485,-8.485) {};
  \node[fill=red] (n16) at (-3.106,-11.591) {};
  \node[fill=red] (nv) at (6.212,-23.182) {$v$};
  \draw[dash pattern=on 3pt off 3pt] (n0) -- (n1);
  \draw (n0) -- (n2);
  \draw (n0) -- (n3);
  \draw (n0) -- (n4);
  \draw[dash pattern=on 3pt off 3pt] (n1) -- (n5);
  \draw (n1) -- (n6);
  \draw (n1) -- (n7);
  \draw (n2) -- (n8);
  \draw (n2) -- (n9);
  \draw (n2) -- (n10);
  \draw (n3) -- (n11);
  \draw (n3) -- (n12);
  \draw (n3) -- (n13);
  \draw (n4) -- (n14);
  \draw (n4) -- (n15);
  \draw (n4) -- (n16);
  \draw[dashed, decorate, decoration={zigzag,amplitude=2mm,segment length=8mm}] (n5) -- (nv);
\end{tikzpicture}
\caption{The distance-2 shell $S_{0}$, as well as the dotted path from $e_0$ to $v$. }
\end{subfigure}
\caption{Traversing the tree}
\label{fig:Btbuilding2}
\end{figure}
\subsection{Distances in the building}

The following lemma describes the set of relevant matrices that we need to focus on to get all vertices $\mathcal{P}_{\mathcal{B}} \subseteq \mathcal{B}$.
\begin{lemma}
  \label{le:Aset}
  Define $\mathcal{A} \subseteq \GL_{3}(F_{\pi})$ as 
  \begin{equation}
      \mathcal{A} =\{\alpha \in \GL_{3}(F_{\pi})  \mid  \alpha^{*}\alpha \in \GL_{3}(\mathcal{O}_{\pi})\}.
  \end{equation}
  Then, for each $v \in \mathcal{P}_\mathcal{B}$ one has $v = \alpha v_{0}$ for some $\alpha \in \mathcal{A}$ where $v_{0} \in \mathcal{B}$ is the origin.
\end{lemma}
\begin{proof}
    It is clear that $v = \{\pi^{i}\Lambda\}_{i \in \mathbb{Z}}$ for some self-dual lattice $\Lambda$. We write $\Lambda = g \mathcal{O}_{\pi}^{3}$ for some $g \in \GL_{3}(F_{\pi})$. Then, the self-duality condition implies that 
    $(g^{*})^{-1} \mathcal{O}_{\pi}^{3} = g \mathcal{O}_{\pi}^{3}$. This means that $g ^{*} g \mathcal{O}_{\pi}^{3} = \mathcal{O}_{\pi}^{3}$.
\end{proof}

Here's something important to know about elements of $\mathcal{A} \subseteq \GL_{3}(F_{\pi})$.
\begin{lemma}
  \label{le:diagonals}
    Write the Cartan decomposition of $g$ as $g = k a k'$ for $k,k' \in \GL_{3}(\mathcal{O}_{\pi})$ and some diagonal $a \in \GL_{3}(F_{\pi})$. Then, we can assume that $a = \diag(\chi^{n},1,\chi^{-n})$ for some $n \in \mathbb{Z}_{\geq 0}$.

\end{lemma}
\begin{proof}
    Since for some $\gamma \in \GL_{3}(\mathcal{O}_{\pi})$, we get $ g^{*} = \gamma g^{-1} =  \gamma (k')^{-1} a^{-1} k^{-1}$, by the uniqueness part of Theorem \ref{th:cartan}, we know that one must take $a= \diag(\chi^{n}, 1 , \chi^{-n})$.
\end{proof}

Consider the group 
$\GL_{3}(F_{\pi})$ and define for $g\in \GL_{3}(F_{\pi})$
\begin{equation}
    l(g) = -2 \min_{i,j}\{ v_\pi (g_{ij})\}.
    \label{eq:defi_of_l}
\end{equation}

\begin{lemma}
  \label{le:ispositive}
  Recall $\mathcal{A}$ from Lemma \ref{le:Aset}.
  \begin{enumerate}
\tightlist
  \item Given any $k, k' \in \GL_{3}(\mathcal{O}_{\pi})$, we have $l(kgk') = l(g)$.
      \item One has $l(g_{1}g_{2}) \leq l(g_{1}) + l(g_{2})$ for any $g_{1}, g_{2} \in \GL_{3}(F_{\pi})$.
	  \item For $g \in \mathcal{A}$, $l(g) \geq 0$ and $l(g^{-1}) = l(g)$.
	  \item For $g,h \in \mathcal{A}$, one has $l(g^{-1} h) \geq 0$. 
	      \item 
	      The map $\tilde{d}(g,h):= \frac{1}{2}[l(g^{-1}h) + l(h^{-1} g)]$ defines a metric on $  \mathcal{A} / \GL_{3}(\mathcal{O}_{\pi})$.
  \end{enumerate}
\end{lemma}
\begin{proof}
See Appendix \ref{proofoflemma5.3}.
\end{proof}

One can use the map $l$ from Equation \ref{eq:defi_of_l} to understand the action of $\U_{3}(F_{\pi})$ on $\mathcal{B}$. 
Here is a lemma about what $l$ can measure.
\begin{lemma}
  \label{le:hecke_nghbs}
  Let $\mathcal{A}$ be as in Lemma \ref{le:Aset} and let $g,h \in \mathcal{A}$.
  Let $\Lambda_{g} = g \mathcal{O}_{\pi}^{3}$ and $\Lambda_{h} = h \mathcal{O}_{\pi}^{3}$ and $\Lambda_{h} \neq \Lambda_{g}$. 
  Thus $\Lambda_{g},\Lambda_{h}$ are distinct self-dual lattices in $F_{\pi}^{3}$.
  Then $\tilde{d}(g,h) = 2$ if and only if $ \Lambda_{g }\subseteq  \pi^{-1}\Lambda_{h}$ and $ \Lambda_{h} \subseteq \pi^{-1} \Lambda_{g}$.
\end{lemma}
\begin{proof}
The proof can be found in Appendix \ref{proofoflemma5.4}
\end{proof}

\begin{proposition}
\label{prop:distance}
    Consider $\mathcal{B}$ as a graph and let $d(x,y)$ denote the edge-distance between two points $x,y \in \mathcal{B}$  
    Let $g,h \in \mathcal{A} \subseteq \GL_{3}(F_{\pi})$. 
    Then, one has $\tilde{d}(g,h) = d(g v_{0},hv_{0})$.
\end{proposition}
% \begin{remark}
  % The graph $\mathcal{B}$ is actually connected. We will see this is in the next section.
%\end{remark}
\begin{proof}
The full proof can be found in Appendix \ref{proofoflemma5.5}.
\end{proof}

\subsection{Proof of main theorem}
\label{ss:proof}
Denote $\Gamma = \U_{3}(\mathbb{Z}[\chi^{-1}])$.
We now describe the overall idea. Let $\mathcal{H} \subseteq \Gamma$ be the finitely generated group of Clifford$+\mathsf{R}$. 
We want to consider the orbit $\mathcal{H} \cdot e_{\mathcal{B}}$ and show that it is the same as the orbit $\Gamma \cdot e_{\mathcal{B}}$. 
We will do this by starting with some $v = g v_{0} $  for some $g \in \Gamma$.  
Via induction on $l(v)$, we will show that $v = h v_{0}$ for some $h \in \mathcal{H}$. Once this is done, we have 
\begin{equation}
    g^{-1}h  \in  \Gamma \cap \U_{3}(\mathcal{O}_{\pi}) = \U_{3}(\mathcal{O}_{F}).
\end{equation}
The following lemma helps us show that all monomials are in Clifford$+\mathsf{R}$.

\begin{proposition}
The set $\U_{3}(\mathcal{O}_{F})$ is the set of monomial matrices. They form a subset $\U_{3}(\mathcal{O}_{F}) \subseteq \mathcal{H}$ and are therefore in the Clifford$+\mathsf{R}$ set.
\end{proposition}
\begin{proof}
    First observe that if $m,n \in \mathbb{Z}$ then
    \begin{equation}
	m^{2} + n^{2} -mn \leq 1  \implies   (m,n) \in \{ (0,0), (\pm 1,0),(0, \pm1), (1, 1), (-1,-1)\}.
    \end{equation}
Hence, if $\alpha = m + \omega n $ satisfies $\alpha \overline{\alpha} \leq 1$, we know that $\alpha \in \{0, \pm 1, \pm \omega, \pm  \omega^{2 }\}$.
    Now, if $a_{1} , a_{2} ,a_{3} \in \mathcal{O}_{F} = \mathbb{Z}[\omega]$ satisfy    
	$a_1\overline{a_{1}} + a_{2} \overline{a_{2}} + a_{3} \overline{a_{3}} = 1,$   
    then we know that each $a_{i} \overline{a_{i}} \in \mathbb{Z}_{\geq 0}$ so exactly one of the three terms on the left must be $1$. 
    Hence, if $A \in \GL_{3} (\mathcal{O}_{F})$ satisfies $A^{*} A = I$, we know that each row has exactly one non-zero element and it must be one of $\{\pm 1, \pm \omega, \pm \omega^{2} \}$.

    Finally, all of these monomials lie in $\mathcal{H}$, since $\mathcal{H}$ contains $S$, $R$, $H^{2}$, and $HSHHHSH$, and one checks that these four elements generate $\U_{3}(\mathcal{O}_{F})$.
\end{proof}

\begin{proof}
{\bf (of Theorem \ref{th:ring_equality_p_is_3})}

    We will argue that for any pure vertex $v \in \mathcal{P}_{\mathcal{B}}$ such that $v \neq e_{0}$, one can find an element $g \in \mathcal{H}$ such that $d(g e_{0}, v) = d(e_{0},v) -2 $. See Figure \ref{fig:Btbuilding2}(b) for a schematic.

    Let $S_{0} \subseteq \mathcal{P}_{\mathcal{B}}$ be the set
    \begin{equation}
        S_{0} =\{w \in \mathcal{P}_{\mathcal{B}} \mid d(e_{0},w)=2\}.
    \end{equation}
    We observe that $\# S_{0} = 4 \times 3 = 12$. 
    Write $e_{1} = H e_{0}$. We have $d(e_{1},e_{0}) = 2$  and $e_{1} \in S_{0}$.
    It is  sufficient to show that $\U_{3}(\mathcal{O}_{F}) e_{1} = S_{0}$. Indeed, at least one $e \in S_{0}$ lies on the shortest path connecting $e_{0}$ to $v$ and if we show $e \in \U_{3}(\mathcal{O}_{F}) e_{1}$, we are done. 
    We also have that $\Stab(e_0,\Gamma) = \U_{3}(\mathcal{O}_{F}) \implies \Stab(e_{1},\Gamma) = H U_{3}(\mathcal{O}_{F})H^{-1}$. 
    So the group
    \begin{equation}
        \mathcal{G} = \{ g \in \Stab(e_{0},\Gamma) \mid g e_{1} = e_{1}\} = H \U_{3}(\mathcal{O}_{F}) H^{-1} \cap \U_{3}(\mathcal{O}_{F}).
    \end{equation}
    By calculating these finite groups explicitly, we know that $\# \mathcal{G} = 108$, whereas $\#\U_{3}(\mathcal{O}_{F}) =  3! \times 6^{3} = 1296$. By the orbit-stabilizer theorem, we get that $\# \U_{3}(\mathcal{O}_{F}) e_{1} = 12$ and hence we are done.
\end{proof}

\section{Acknowledgments}

We would like to thank Shai Evra, Ori Parzanchevski for discussing and sharing their understanding about Bruhat-Tits buildings with us. The authors also thank the QPL reviewers for their comments. This work was supported in part by Canada’s NSERC, NTT Research, and the Royal Bank of Canada. IQC and the Perimeter Institute (PI) are supported in part by the Government of Canada through ISED and the Province of Ontario (PI). JY acknowledges support from the NSERC project FoQaCiA under Grant No. ALLRP-569582-21. NG acknowledges funding
from the ERC Grant 101096550 and from the SNSF grant 225437.
\\
\\
An LLM was used to proofread the completed manuscript and identify potential errors, which were subsequently reviewed and addressed by the authors. It was also used to assist with writing the TikZ code for the figures.
\bibliographystyle{eptcs}
\bibliography{refs}

\appendix

\section{Proof of Proposition \ref{pr:simplex_classification}}
\label{ap:simplex_classification}
\begin{definition}
    Given a self-dual lattice chain $\{\Lambda_i\}_{i \in \mathbb{Z}}$, we define the anchor of the lattice chain  as the unique half-integer $i \in \tfrac{1}{2}\mathbb{Z}$ defined as the following.
\begin{enumerate}
    \item If $\Lambda_j = \dual{\Lambda_j}$ for some $j \in \mathbb{Z}$, we denote the anchor to be $i=j$.
    \item If $\Lambda_j = \dual{\Lambda_{j+1}}$ for some $j \in \mathbb{Z}$, we denote the anchor to be $i=j+1/2$. 
\end{enumerate}
\end{definition}

We assume the labeling of a lattice chain $\{\Lambda_i\}_{i \in \mathbb{Z}}$ is such that the anchor point is at $i=0$ or $i=\frac{1}{2}$ for each case respectively.

\begin{proposition}
    \label{pr:simplex_classification_appendix}
For $\mathcal{B}$ described above, we have the following classification of simplices.
  \begin{enumerate}
    \item All the $0$-simplices  of $\mathcal{B}$ are one of the following two types.
      \begin{itemize}

          \item 
Lattice chains of the form $\{ \pi^{i}\Lambda\}_{i \in \mathbb{Z}}$ for some self-dual lattice $\Lambda \subseteq F_{\pi}^{3}$  arranged as 
\begin{equation}
    \dots \subseteq \pi^{2} \Lambda \subseteq \pi \Lambda \subseteq \Lambda \subseteq \pi^{-1} \Lambda \subseteq \dots
\end{equation}
    \item 
	Lattice chains of the form $\{ \pi^i\Lambda, \pi^i \dual{\Lambda}\}_{i \in \mathbb{Z}}$ 
	 such that $\dual{\Lambda} \subsetneq\Lambda$ and
	 \begin{equation}
	       \dots \subseteq \pi \Lambda \subseteq \dual{\Lambda}  \subseteq \Lambda \subseteq \pi^{-1} \dual{\Lambda}  \subseteq \dots
	 \end{equation}
      \end{itemize}
  \item All the 1-simplices of $\mathcal{B}$ are of the form $\{ \pi^i\Lambda_1, \pi^i \Lambda_0, \pi^i \dual{\Lambda_{1}} \}_{i \in \mathbb{Z}}$  for some lattices $\Lambda_0,\Lambda_1\subseteq F_{{\pi}}^3$ 
      such that $\Lambda_0$ is self dual and $\Lambda_1^{\dual{}} \subsetneq \Lambda_1$ arranged as
      \begin{equation}
              \dots \subseteq \pi \Lambda_0 \subseteq \pi\Lambda_{1}  \subseteq \dual{\Lambda_{1}} \subseteq \Lambda_0 \subseteq \Lambda_{1} \subseteq \pi^{-1} \dual{\Lambda_{1}} \subseteq \dots
      \end{equation}
      
      \item There are no 2-simplices in $\mathcal{B}$.
      \end{enumerate}
\end{proposition}
\begin{proof}
\begin{figure}
    \centering
    \label{fig:alternatingvert}
    \begin{tikzpicture}[
        node distance=6cm,
        dot/.style={circle, fill=black, inner sep=3pt},
        edgelabel/.style={font=\scriptsize, inner sep=4pt, outer sep=0pt},
        lababove/.style={font=\footnotesize, align=center, above=5pt},
    ]

    \node[dot] (left) {};
    \node[dot, right=of left] (mid) {};
    \node[dot, right=of mid] (right) {};

    \node[lababove] at (left.north)  {$\{\pi^i \Lambda_0\}$};
    \node[lababove] at (mid.north)   {$\{\pi^i \Lambda_1, \pi^i \dual{\Lambda_1}\}$};
    \node[lababove] at (right.north) {$\{\pi^i \Lambda_2\}$};

    \draw[line width=1.5pt] (left) -- node[edgelabel, above] 
        {$\{\pi^i \Lambda_1,\, \pi^i \Lambda_0,\, \pi^i \dual{\Lambda_1}\}$} (mid);

    \draw[line width=1.5pt] (mid) -- node[edgelabel, above] 
        {$\{\pi^i \Lambda_1,\, \pi^i \Lambda_2,\, \pi^i \Lambda_1^\sharp\}$} (right);

    \end{tikzpicture}
    \caption{An alternating vertex (centre) and two of its four pure neighbours}
\end{figure}
We will systematically investigate lattice chains of rank $1,2,3$ to get all possible self-dual lattice chains.
    \begin{enumerate}[(a)]
    \item 
    { Let $\{\Lambda_i\}$ be a rank-1 lattice chain in $\mathcal{B}$. Hence $\Lambda_0$ must be $\pi$-equivalent to every lattice in the chain. In particular, $\Lambda_1 = \pi^{-k} \Lambda_0$ for some $k \in \mathbb{Z}_{\geq 1}$. To maintain rank 1, the only possibility is $k=1$.
    
    By assumption the anchor point of the chain is either at $0$ or $\tfrac{1}{2}$. If the anchor point was at index $\frac{1}{2}$, it would imply that $\dual{\Lambda_0} = \Lambda_1 $ which would be an instance of odd-self duality, contradicting Lemma \ref{le:dual_scaling}. The only remaining option would be for the anchor point to be at index $0$, which implies self-duality for $\Lambda_0$. 
    
    Clearly $\{ \pi^i\Lambda_0\}_{i \in \mathbb{Z}}$ cannot contain any subchains and will be a $0$-simplex in $\mathcal{B}$.
Setting $\Lambda= \Lambda_0$ gives the chain equation.

    }
    \item Let $\{\Lambda_i\}$ be a rank-2 lattice chain in $\mathcal{B}$. 
	\begin{itemize}

	    \item 

        { First, let us assume that the anchor point is at index 0. Hence, we get that the chain $\{\pi^{i}\Lambda_{0}\} \subseteq \{\Lambda_{i}\}$ 
	forms a self-dual subchain inside $\{\Lambda_{i}\}$. 
	But since $\{\Lambda_{i}\}$ is a rank $2$ lattice chain, 
        there exists another $\pi$-equivalence class of lattices. $\Lambda_1$ could not be $\pi$-equivalent to $\Lambda_0$, as this would cause the whole chain to be $\pi$-equivalent to $\Lambda_0$, contradicting the chain being rank $2$. Additionally, $\Lambda_1$ could not be $\pi$-equivalent to a self dual lattice as there can only be one self-dual lattice in the chain, because there can be at most one anchor point in the chain. 
        By the anchor being at index $0$ we have $\Lambda_{-1} = \dual{\Lambda_1}$, while rank $2$ forces $\Lambda_{-1} = \pi \Lambda_1$ (a rank-$2$ chain has period $2$). Hence $\dual{\Lambda_1} = \pi \Lambda_1$, i.e.\ $\Lambda_1$ is $\pi$-equivalent to its dual with the odd exponent $i=1$, contradicting Lemma \ref{le:dual_scaling}.
        
}

\item 

Suppose now that $\{\Lambda_{i}\}$ has its anchor at $\frac{1}{2}$. This directly implies that $\dual{\Lambda_0} = \Lambda_1$. Hence, $\Lambda_0$ and $\Lambda_1$ cannot be $\pi$-equivalent (Lemma \ref{le:dual_scaling}), so all other lattices in the chain must be scalings of those 2. Since there will be no self-dual lattice in the chain, the lattice chain $\{\Lambda_i\}_{i \in \mathbb{Z}}$ 
cannot contain any rank 1 0-simplex lattice chains. Alternating between scalings of $\Lambda_0$ and $\Lambda_1$ is forced due to the lattice chain being fixed by $\pi$ scalings, which leads us to the second type of 0-simplex in the statement. 
	\end{itemize}
    \item Let $\{\Lambda_i\}$ be a rank-3 lattice chain in $\mathcal{B}$. 

	Because duality is an involution, at least one of the three $\pi$-equivalence classes of lattices in $\left\{\Lambda_{i}\right\}$ contains a self-dual lattice. 
    Let $\Lambda_{0}$ be this self-dual lattice (and so anchor is at 0). There can be no other $\pi$-equivalence class containing a self-dual lattice other than $\{\pi^i \Lambda_0\}_{i \in \mathbb{Z}}$ because of the construction of the anchor point. In particular $\dual{\Lambda_1}$ must be $\Lambda_{-1}$.
	
	Hence, we must necessarily have that this lattice chain is of the type given in the statement and is a $1$-simplex since it contains the two types of $0$-simplices demonstrated above.

	\item There are no rank 4 or above lattices chains since we are working with lattices in $F_{\pi}^3$. Hence, all self-dual lattice chains are classified and no $2$-simplices could be found.
\end{enumerate}
\end{proof}
\section{Proof of Proposition \ref{pr:forest}}
\label{se:forest_proof}
\begin{proposition}
    The $0$-simplices and $1$-simplices of $\mathcal{B}$ form a forest. That is, $\mathcal{B}$ has no loops.
\end{proposition}
\begin{proof}
    Suppose that $\mathcal{B}$ contains a loop (see Figure \ref{fig:loop}). Since each edge connects a vertex in $P_\mathcal{B}$ to a vertex in $A_\mathcal{B}$, we may arbitrarily choose a vertex $\{\pi^iS_0\}_{i \in \mathbb{Z}} \in P_\mathcal{B}$ as the start and end of the loop. This loop, as for any other path, can be described as a sequence of edges.

    \begin{align}
    \dots \subset \pi A_0\subset \pi S_0 \subset A_0 \subset S_0 \subset A_0^{\dual{}} \subset \pi^{-1}S_0 \subset \pi^{-1} A_0^{\dual{}} \subset \dots \\
    \dots \subset \pi A_0\subset \pi S_1\subset A_0 \subset S_1 \subset A_0^{\dual{}} \subset \pi^{-1}S_1 \subset \pi^{-1} A_0^{\dual{}} \subset \dots\\
    \dots \subset \pi A_1\subset \pi S_1 \subset A_1 \subset S_1 \subset A_1^{\dual{}} \subset \pi^{-1}S_1 \subset \pi^{-1} A_1^{\dual{}} \subset \dots\\
    \dots \subset \pi A_1\subset \pi S_2 \subset A_1 \subset S_2 \subset A_1^{\dual{}} \subset \pi^{-1}S_2 \subset \pi^{-1} A_1^{\dual{}} \subset \dots\\
    {\vdots}\\
    \dots \subset \pi A_n\subset \pi S_n \subset A_n \subset S_n \subset A_n^{\dual{}} \subset \pi^{-1}S_n \subset \pi^{-1} A_n^{\dual{}} \subset \dots \\
    \dots \subset \pi A_n\subset \pi S_0 \subset A_n \subset S_0 \subset A_n^{\dual{}} \subset \pi^{-1}S_0 \subset \pi^{-1} A_n^{\dual{}} \subset \dots \\
    \end{align}

\begin{figure}
    \centering
    \begin{center}
\begin{tikzpicture}[
    dot/.style={circle, fill=black, inner sep=2.5pt},
    >={Stealth[length=3mm, width=2mm]},
    every label/.append style={font=\small},
]

    \node[dot] (origin) at (90:3cm) {};
    \node[dot] (n1)     at (30:3cm) {};
    \node[dot] (n2)     at (-30:3cm) {};
    \node at (0,-1.5) {\large $\dots$};
    \node[dot] (n5)     at (210:3cm) {};
    \node[dot] (n6)     at (150:3cm) {};

    \draw[->, thick, bend left=20]
        (origin) to node[above right, xshift=5pt, yshift=2pt] {$\{A_0, S_0, A_0^\sharp\}$} (n1);

    \draw[->, thick, bend left=20]
        (n1) to node[below right, xshift=6pt, yshift=-2pt] {$\{A_0, S_1, A_0^\sharp\}$} (n2);

    \draw[->, thick, bend left=20]
        (n5) to node[below left, xshift=-6pt, yshift=-3pt] {$\{A_{n-1}, S_n, A_{n-1}^\sharp\}$} (n6);

    \draw[->, thick, bend left=20]
        (n6) to node[above left, xshift=-6pt, yshift=3pt] {$\{A_n, S_0, A_n^\sharp\}$} (origin);

\end{tikzpicture}
\caption{What a loop would look like in $\mathcal{B}$}
\label{fig:loop}
\end{center}
\end{figure}

    by starting at $A_0$ and forming a chain of subsets on the left diagonal, then starting at $\dual{A_0}$ and forming a chain of subsets on the right diagonal you arrive at the following relations:

    \begin{align}
    \pi^{n}A_n \subset A_0 \\
        A_0^{\dual{}} \subset \pi^{-n}A_n^{\dual{}}
    \end{align}

    this relationship implies that you can form the following self-dual lattice chain

    \begin{equation}
        \dots \subset (\pi^n A_n) \subset A_0 \subset S_0 \subset A_0^{\dual{}} \subset (\pi^n A_n)^{\dual{}} \subset \dots
    \end{equation}

    Since the loop uses two distinct alternating vertices $A_0$ and $A_n$ together with the pure vertex $S_0$, this is a self-dual lattice chain of rank greater than $3$ (a $2$-simplex), contradicting Proposition \ref{pr:simplex_classification}.
\end{proof}

\section{Proof of Lemma \ref{antisym}}
\label{proofappendixc}
\begin{lemma}
Consider $\mathbb{F}_{3}^{3}$ and an anti-symmetric matrix $A \in \M_{3}(\mathbb{F}_{3})$ given by the following for some $(a,b) \in \mathbb{F}_{3}^{2} \setminus \{(0,0)\}$:
\begin{equation}
    A 
    = 
    \begin{pmatrix}
	\ &  a & b \\ 
	-a & & \\ 
	-b & & 
    \end{pmatrix}.
\end{equation}

Consider the bilinear form $\langle\ , \ \rangle_{A}$ given by 
\begin{equation}
    \langle x,y\rangle_{A}  = \langle A x, y \rangle.
\end{equation}

Then, the vector $v=(0,b,-a)^{T} \in \mathbb{F}_{3}^{3}$ is the unique vector up to scaling satisfying $\langle v ,w \rangle_{A} = 0$ for all $w \in \mathbb{F}_{3}^{3}$. Furthermore, there are exactly four subspaces $V \subseteq \mathbb{F}_{3}^{3}$ of $\mathbb{F}_{3}$-dimension 2 such that $\mathbb{F}_{3} \cdot v \subseteq V$ and $V^{\perp} = \{ w \in \mathbb{F}_{3}^{3} \mid \langle u,w \rangle_{A} = 0 \text{ for all }u \in V\} = V$.
\end{lemma}

\begin{proof}
    We know that for $x=(x_{1},x_{2},x_{3})$ and $y=(y_{1},y_{2},y_{3})$, one has
\begin{equation}
    \langle x,y \rangle_{A} = a(x_{1} y_{2} - x_{2} y_{1}) + b (x_{1} y_{3} - x_{3} y_{1}).
\end{equation}
Then, it is clear that the only way one can have $\langle x, y \rangle_{A} = 0$ for all $y \in \mathbb{F}_{3}^{3}$ is when $x_{1} = 0$ and $a x_{2} + b x_{3} = 0$. This shows that $v = (0,b,-a)$ is the vector as required.

Then assuming $a \neq 0$, one can identify $(x_{1}, x_{2},x_{3}) + \mathbb{F}_{3} \cdot v \leftrightarrow (x_{1},x_{2})$ and get an isomorphism of $\mathbb{F}_{3}$-spaces ${\mathbb{F}_{3}^{3}}/{ \mathbb{F}_{3} \cdot v } \simeq \mathbb{F}_{3}^{2}$. Because of the property of $v$, the map $\langle \ , \ \rangle_{A}$ quotients on $\mathbb{F}_{3}^{3} / \mathbb{F}_{3} \cdot v$ to the map
\begin{equation}
    ( (x_{1},x_{2}), (y_{1},y_{2}))  \mapsto  a(x_{1} y_{2} - x_{2} y_{1}).
\end{equation}
% MAIN ISSUE CORRECTION
This is a non-degenerate alternating form on $\mathbb{F}_{3}^{2}$, so every line $\mathbb{F}_{3} w$ is self-dual: it is isotropic because the form is alternating, and $(\mathbb{F}_{3}w)^{\perp}$ is again $1$-dimensional by non-degeneracy, hence $(\mathbb{F}_{3}w)^{\perp} = \mathbb{F}_{3}w$. There are four such lines $w$ in $\mathbb{F}_{3}^{2}$, giving the four subspaces $V = \mathbb{F}_{3} v + \mathbb{F}_{3} w$.

We leave the case of $a=0$ for the reader which is identical.
\end{proof}

\section{Proof of Lemma \ref{le:techincal_lemma}}
\label{proofoflemma3.16}
\begin{lemma}
	  Let $x,y \in F_{\pi}$ be such that $\overline{x} x + \overline{y} y \in \mathcal{O}_{\pi}$. Then, $x,y \in \mathcal{O}_{\pi}$.
	\end{lemma}
\begin{proof}
    If $v_{\pi}(x) < v_{\pi}(y)$ then we clearly have that $ 0 \leq v_{\pi}( \overline{x} x + \overline{y} y )  = 2 v_{\pi}(x)$ and we are done.
    
    Otherwise, let's assume that $x = \chi^{n}x', y=\chi^{n}y'$ for some $n < 0$ such that $x',y' \in \mathcal{O}_{\pi}^{\times}$.
    Then, we have $x',y' \in \pm 1 + \pi \mathcal{O}_{\pi}$ implying that $\overline{x'} x' + \overline{y'} y' \in  2 + \pi \mathcal{O}_{\pi}$. Therefore, we must have that $x \overline{x} + y \overline{y} \in \chi^{n} \overline{\chi}^{n}(2 + \pi \mathcal{O}_{\pi} )$ forcing $n$ to be non-negative. This is impossible. So $x,y \in \mathcal{O}_{\pi}$.
\end{proof}

\section{Proof of Lemma \ref{le:nondegen2}}
\label{proofofprop5.19}
\begin{lemma}
  \label{le:nondegen2repeat}
  Consider the form $\langle \ , \ \rangle + \mathcal{O}_{\pi}$ defined in 
  Equation \ref{eq:defi_of_angle4}. Then, after identification of $\Lambda / \pi \Lambda \simeq \mathbb{F}_{3}^{3}$, the map $\chi \langle \ , \ \rangle + \pi \mathcal{O}_{\pi}$ is anti-symmetric and is equivalent to the bilinear map $\mathbb{F}_{3}^{3} \times \mathbb{F}_{3}^{3} \rightarrow  \mathbb{F}_{3}$ considered in Lemma \ref{antisym}.
\end{lemma}
\begin{proof}
    To show anti-symmetric, one sees that the unit $\overline{\chi}/\chi = (3-\chi)/\chi = -1 + \overline{\chi} \in \mathcal{O}_{\pi}^{\times}$.
    Therefore, one gets that $\overline{\chi}/\chi  = -1 \pmod{ \pi \mathcal{O}_{\pi}}$. So one sees that for $x,y \in \Lambda$ 
    we get 
    \begin{align}
	\chi \langle x, y \rangle + \pi \mathcal{O}_{\pi}  & =  \chi \overline{\langle y, x \rangle } + \pi \mathcal{O}_{\pi}. \\
							   & =  ({\chi}/{\overline{\chi}}) \cdot \left( \overline{\chi} 
\overline{\langle y, x \rangle } + \pi \mathcal{O}_{\pi} \right).
    \end{align}
    Then, using that for $a \in \mathcal{O}_{\pi}$ one has $\overline{a} = a \pmod{\pi \mathcal{O}_{\pi}}$, we get the required anti-symmetric property.

    Now suppose that $\Lambda = g \mathcal{O}_{\pi}^{3}$ for some $g \in \GL_{3}({F}_{\pi})$. 
Without loss of generality, we can replace $g$ by $gk$ for $k \in \GL_{3}(\mathcal{O}_{\pi})$. 
Hence, one can write $g$ in Hermite normal form (up to scaling and rescaling by some power of $\chi$) 
and assume that $g$ is a lower-triangular matrix.
    
Using the isomorphism $x \mapsto  g^{-1}x$, one can get the isomorphism
    \begin{equation}
	\frac{\Lambda}{\pi \Lambda} = \frac{g \mathcal{O}_{\pi}^{3}}{ \pi g \mathcal{O}_{\pi}^{3}} \rightarrow    \frac{\mathcal{O}_{\pi}^{3} }{\pi \mathcal{O}_{\pi}^{3}}.
    \end{equation}
    Keeping track of the bilinear product, one knows that the relevant form on $\mathcal{O}_{\pi}^{3} / \pi \mathcal{O}_{\pi}^{3} \simeq \mathbb{F}_{3}$ is 
    \begin{equation}
	(x,y) \mapsto  \chi \langle g x, gy \rangle + \pi \mathcal{O}_{\pi} \in \mathbb{F}_{3}.
    \end{equation}

    Let 
    \begin{equation}
        g = 
	\begin{pmatrix}
	    g_{11} & & \\ 
	    g_{21} & g_{22} & \\ 
	    g_{31} & g_{32} & g_{33}
	\end{pmatrix}.
    \end{equation}
    Then we get that the matrix
    \begin{equation}
        g^{*} g =  
	\begin{pmatrix}
	    \overline{g_{11}}g_{11} + 
	    \overline{g_{21}}g_{21} + 
	    \overline{g_{31}}g_{31}  
	    & 
 \overline{g_{21}} g_{22} 
 + \overline{g_{31}} g_{32} 
	    & 
	    \overline{g_{31}} g_{33} 
	     \\ 
	     \overline{g_{22}}g_{21} +  \overline{g_{32}}g_{31}
	    & 
	    \overline{g_{22}} g_{22} + \overline{g_{32}}g_{32}
	    & 
\overline{g_{32}} g_{33}
	     \\ 
	    \overline{g_{33}} g_{31} 
	    & 
	    \overline{g_{33}} g_{32}
	    & 
	    \overline{g_{33}} g_{33}
	\end{pmatrix}
    \end{equation}
    must satisfy $ \langle \chi g^*g x ,  y \rangle \in \mathcal{O}_{\pi}$ for all $x,y \in \mathcal{O}_{\pi}^{3}$. 
    Therefore, all the entries must be in $\mathcal{O}_{\pi}$. 
    In particular, the diagonal entries are fixed by $\overline{( \ )}$ and lie in $\mathbb{Z}_{3}$. 
    Therefore the valuation $v_\pi(\ )$ evaluated at the diagonal elements of $\chi g^{*}g$ must be at least 1. Hence all the diagonals are in $\pi \mathcal{O}_{\pi}$ and therefore reduce to $0 \in \mathbb{F}_{3}$.
    
    This implies that $g_{33} \in \mathcal{O}_{\pi}$. 
    Furthermore, the entry $\overline{g_{22}} g_{22} + \overline{g_{32}}g_{32} \in \mathcal{O}_{\pi}$. By Lemma \ref{le:techincal_lemma} applied to this entry, $g_{22} , g_{32} \in \mathcal{O}_{\pi}$ (and $g_{33}\in\mathcal{O}_{\pi}$ from above). In particular, we get that the bottom-right half of $\chi g^{*}g = 0 \pmod{\pi \mathcal{O}_{\pi}}$. 

    Depending on $g_{21},g_{31} \in F_{\pi}$, the only non-zero entries in $\chi g^{*} g$ can be 
    $a = \chi( \overline{g}_{21} g_{22} + \overline{g_{31}} g_{32})$ and
    $b=\chi^{}(\overline{g_{31}}  )g_{33}$. Observe that both $a,b$ cannot be zero, otherwise one gets that $\dual{\Lambda} = \Lambda$.

\end{proof}
\section{Proof of Proposition \ref{pr:chain}}
\label{proofofprop5.22}
\begin{proposition}
\label{pr:chain_repeat}
    Let $g,h \in \GL_{3}(F_{\pi})$ be such that $\Lambda_{g} = g \mathcal{O}_{\pi}^{3}$ and $\Lambda_{h} =h \mathcal{O}_{\pi}^{3}$ are self-dual lattices. Furthermore, we assume that 
    \begin{equation}
	\dots \subseteq \pi^{n}\Lambda_{g} \subseteq \Lambda_{h} \subseteq \pi^{-n}\Lambda_{g} \subseteq  \pi^{-2n} \Lambda_{h} \subseteq \dots,
	\label{eq:non-lattice-chain-repeat}
    \end{equation}
    where $n \in \mathbb{Z}_{\geq 0}$ is the least such $n$ that makes Equation \ref{eq:non-lattice-chain} hold.
    Then, there exists a sequence of self-dual lattices $\Lambda_{1}, \Lambda_{2} ,\dots, \Lambda_{n-1} \subseteq F_{\pi}^{3}$ such that 
    \begin{align}
	\label{eq:non-lattice-chain2-repeat}
	\dots \subseteq \pi \Lambda_{1} \subseteq &  \Lambda_{g} \subseteq  \pi^{-1} \Lambda_{1} \subseteq \pi^{-2}\Lambda_{g} \subseteq \dots \\
	 \dots \subseteq \pi \Lambda_{2} \subseteq &  \Lambda_{1} \subseteq  \pi^{-1} \Lambda_{2} \subseteq \pi^{-2}\Lambda_{1} \subseteq \dots \\
						   & \vdots \\
	 \dots \subseteq \pi \Lambda_{n-1} \subseteq &  \Lambda_{h} \subseteq  \pi^{-1} \Lambda_{n-1} \subseteq \pi^{-2}\Lambda_{h} \subseteq \dots \\
    \end{align}
\end{proposition}
\begin{proof}
We claim that we can find a basis of $v_{1},v_{2},v_{3} \in F_{\pi}^{3}$ such that $\Lambda_{g} = v_{1} \mathcal{O}_{\pi} + v_{2} \mathcal{O}_{\pi} + v_{3} \mathcal{O}_{\pi}$ and $\Lambda_{h} = \chi^{n_{1}} v_{1} \mathcal{O}_{\pi} +  \chi^{n_{2}}v_{2} \mathcal{O}_{\pi} + \chi^{n_{3}} v_{3} \mathcal{O}_{\pi}$ for some $n_{1} \geq n_{2} \geq n_{3} \in \mathbb{Z}$. 

    Here is how one can get this claim. 
    Observe that $g^{-1} h \in \GL_{3}(F_{\pi})$ admits a Cartan decomposition.
    Hence for some diagonal $a= \diag( \chi ^{n_{1}}, \chi^{n_{2}}, \chi^{n_{3}})$ and $k_{1},k_{2} \in \GL_{3}(\mathcal{O}_{\pi})$, one has $g^{-1} h = k_{1} a k_{2} \implies  (g k_{1})^{-1} h = a k_{2}$. 
    Note that $g \mathcal{O}_{\pi}^{3} = (gk_{1}) \mathcal{O}_{\pi}^{3}$ because $k_{1} \in \GL_{3}(\mathcal{O}_{\pi})$.
    We observe that the columns $v_{1},v_{2},v_{3}$ of $gk_{1} \in \GL_{3}(F_{\pi})$ are the required basis. Indeed, $\Lambda_{h} = gk_{1}(gk_{1})^{-1} h \mathcal{O}_{\pi}^{3} = (gk_{1}) a k_{2} \mathcal{O}_{\pi}^{3} = (gk_{1}) a \mathcal{O}_{\pi}^{3}$ which is exactly what we need.

    Denote $g_{1} = g k_{1}$.
    Because of self-duality of $\Lambda_{g},\Lambda_{h}$, 
    we also know that $(g_{1}^{*})^{-1} = g_{1} l_1$ and $(h^{-1})^{*} = h l_2$ for some $l_1,l_2 \in \GL_{3}(\mathcal{O}_{\pi})$. 
    This implies that $((g_{1}^{-1} h)^{*})^{-1} = l_{1} ((k_{1} a k_{2})^{*})^{-1} l_{2}$. By the uniqueness of the Cartan decomposition, we must therefore have $n_{1}=-n_{3}$ and $n_{2}=0$. 
    Because $ \pi^{n} \Lambda_{h} \subseteq \Lambda_{g} \subseteq \pi^{-n} \Lambda_{h}$, we have $ n \geq n_{1} \geq n_{2} \geq n_{3} \geq -n$. By the minimality of $n$, we must have $(n_{1},n_{2},n_{3})=(n,0,-n)$.

    Consider the Gram matrix $A_{g} = [\langle v_{i}, v_{j}\rangle]_{1 \leq i,j \leq 3}$ of $\Lambda_{g}$. Because $\Lambda_{g}$ is self-dual, we must have that $A_{g} \in \GL_{3}(\mathcal{O}_{\pi})$ by Lemma \ref{le:gramm_matrix}.
    The Gram matrix $A_{h}$ of $\Lambda_{h}$ must then be 
    \begin{equation}
        \begin{pmatrix}
	    (\chi \overline{\chi})^{n}\langle v_{1} , v_{1}\rangle &  \chi^{n}\langle v_{1}, v_{2}\rangle & \langle v_{1} ,v_{3}\rangle \\
	    \overline{\chi}^{n}\langle v_{2} , v_{1}\rangle & \langle v_{2}, v_{2}\rangle & \chi^{-n }\langle v_{2} ,v_{3}\rangle \\
	    \langle v_{3} , v_{1}\rangle &  \overline{\chi}^{-n}\langle v_{3}, v_{2}\rangle & (\chi \overline{\chi})^{-n} \langle v_{3} ,v_{3}\rangle \\
        \end{pmatrix} \in \GL_{3}(\mathcal{O}_{\pi}).
    \end{equation}
    This means that $\langle v_{2}, v_{3} \rangle \in \pi^{n} \mathcal{O}_{\pi}$ and $\langle v_{3}, v_{3}\rangle \in \pi^{2n} \mathcal{O}_{\pi}$ (the $(2,3)$ entry carries the factor $\chi^{-n}$ of $\pi$-valuation $-n$, while the $(3,3)$ entry carries $(\chi\overline{\chi})^{-n}=3^{-n}$ of $\pi$-valuation $-2n$). To ensure that $\det A_{g} \in \mathcal{O}_{\pi}^{\times}$, we must have $\langle v_{1},v_{3} \rangle , \langle v_{2},v_{2}\rangle \in \mathcal{O}_{\pi}^{\times}$. One can see this via projection of $A_{g}$ under the map $\GL_{3}(\mathcal{O}_{\pi})\rightarrow  \GL_{3}(\mathbb{F}_{3})$.

    It is clear that the lattices $L_{i} = \chi^{i} v_{1} \mathcal{O}_{\pi} + v_{2} \mathcal{O}_{\pi} + \chi^{-i} v_{3} \mathcal{O}_{\pi}$ for $i=1,\dots,n-1$ satisfy Equation \ref{eq:non-lattice-chain2-repeat}. The only thing left to verify is that they are self-dual. For this, by Lemma \ref{le:gramm_matrix}, it is sufficient to argue that the Gram matrices of $L_{i}$ satisfy
    \begin{equation}
        \begin{pmatrix}
	    (\chi \overline{\chi})^{i}\langle v_{1} , v_{1}\rangle &  \chi^{i}\langle v_{1}, v_{2}\rangle & \langle v_{1} ,v_{3}\rangle \\
	    \overline{\chi}^{i}\langle v_{2} , v_{1}\rangle & \langle v_{2}, v_{2}\rangle & \chi^{-i }\langle v_{2} ,v_{3}\rangle \\
	    \langle v_{3} , v_{1}\rangle &  \overline{\chi}^{-i}\langle v_{3}, v_{2}\rangle & (\chi \overline{\chi})^{-i} \langle v_{3} ,v_{3}\rangle \\
        \end{pmatrix} \in \GL_{3}(\mathcal{O}_{\pi}).
    \end{equation}
\end{proof}
\section{Proof of Lemma \ref{le:ispositive}}
\label{proofoflemma5.3}
\begin{lemma}
  \label{le:ispositive_repeat}
  Recall $\mathcal{A}$ from Lemma \ref{le:Aset}.
  \begin{enumerate}

  \item Given any $k, k' \in \GL_{3}(\mathcal{O}_{\pi})$, we have $l(kgk') = l(g)$.
      \item One has $l(g_{1}g_{2}) \leq l(g_{1}) + l(g_{2})$ for any $g_{1}, g_{2} \in \GL_{3}(F_{\pi})$.
	  \item For $g \in \mathcal{A}$, $l(g) \geq 0$ and $l(g^{-1}) = l(g)$.
	  \item For $g,h \in \mathcal{A}$, one has $l(g^{-1} h) \geq 0$. 
	      \item 
	      The map $\tilde{d}(g,h) :=  \frac{1}{2}[l(g^{-1}h) + l(h^{-1} g)]$ defines a metric on $  \mathcal{A} / \GL_{3}(\mathcal{O}_{\pi})$.
  \end{enumerate}
\end{lemma}

\begin{proof}
    \begin{enumerate}

    \item 
    For this statement, observe that for any $a_{1}, a_{2}, a_{3} \in \mathcal{O}_{\pi}$ and $b_{1},b_{2},b_{3} \in F_{\pi} $, one has 
    \begin{equation}
	\min_{i=1,2,3} v_{\pi}(b_{i}) \leq v_{\pi}(a_{1} b_{1} + a_{2} b_{2} + a_{3} b_{3}).
    \end{equation}
    This proves that $l(k g k') \leq l(g)$ for each $k,k' \in \U_{3}(\mathcal{O}_{\pi})$. Replacing $g$ with $k g k'$ and $k,k'$ with their respective inverses gives us the equality.

    \item 
    It is known by Cartan decomposition that any $g \in \GL_{3}(F_{\pi})$ can be written as 
    \begin{equation}
	g = k a k',
    \end{equation}
    where $k,k' \in \GL_{3}(\mathcal{O}_{\pi})$ and $a \in \GL_{3}(F_{\pi})$ is a diagonal matrix with entries that are powers of $\pi$. 
    Writing $g_{1} = k_{1} a_{1} k_{1}'$ and $g_{2} = k_{2} a_{2} k_{2}'$ in their respective Cartan decompositions for $k_{1},k_{2},k_{1}',k_{2}'\in \GL_{3}(\mathcal{O}_{\pi})$ and diagonal $a_{1},a_{2} \in \GL_{3}(F_{\pi})$ gives us that it is sufficient to show
    \begin{equation}
        l(a_{1} k_{1}' k_{2} a_{2}) \leq l(a_{1}) + l(a_{2}).
    \end{equation}
    Note that $l(a_{2}) = l(k_{1}' k_{2} a_{2})$. So the inequality follows from the claim that for any $g \in \GL_{3}(F_{\pi})$ and any diagonal $a \in \GL_{3}(F_{\pi})$, one has 
    \begin{equation}
        l(a g) \leq l(a) + l(g).
    \end{equation}
    This is easy to verify from Equation \ref{eq:defi_of_l}.
        \item 

    For $g \in \mathcal{A}$ we have $g^{*}g \in \GL_{3}(\mathcal{O}_{\pi})$, so $\det(g^{*}g) \in \mathcal{O}_{\pi}^{\times}$ and $v_{\pi}(\det g) = \tfrac{1}{2}v_{\pi}(\det(g^{*}g)) = 0$. If every entry satisfied $v_{\pi}(g_{ij}) \geq 1$, then $g \in \chi\,\M_{3}(\mathcal{O}_{\pi})$ and $v_{\pi}(\det g) \geq 3$, a contradiction. Hence $\min_{i,j} v_{\pi}(g_{ij}) \leq 0$, so $l(g) = -2\min_{i,j} v_{\pi}(g_{ij}) \geq 0$.

    The statement $l(g) = l(g^{-1})$ is a simple consequence of Lemma \ref{le:diagonals}.

    \item 
	We get that $l(h) \leq l(g) + l(g^{-1}h)$ and $l(g) = l (g^{-1}) \leq l(g^{-1}h) + l(h^{-1}) $. In particular, $l(g^{-1} h) \geq \max\{ l(g)-l(h), l(h)-l(g)\} \geq 0 $.
\item One knows that $\tilde{d} \geq 0$, is symmetric and satisfies the triangle inequality. The only thing to see is what happens if $\tilde{d}(g,h) =0$ for some $g,h \in \mathcal{A}$. In this case, we know that $l(g^{-1} h) = 0 = l(h^{-1}g)$.
    
    If the Cartan decomposition of $g^{-1} h$ has the diagonal $a = \diag(\chi^{n_{1}},\chi^{n_{2}},\chi^{n_{3}})$ with $n_{1} \geq n_{2} \geq n_{3}$, then $l(g^{-1}h), l(h^{-1}g)$ force that $0 \geq n_{1} \geq n_{2} \geq n_{3} \geq 0$.
    This means that $g \in h \GL_{3}(\mathcal{O}_{\pi})$ and we are done.

    \end{enumerate}
\end{proof}
\section{Proof of Lemma \ref{le:hecke_nghbs}}
\label{proofoflemma5.4}
\begin{lemma}
  Let $\mathcal{A}$ be as in Lemma \ref{le:Aset} and let $g,h \in \mathcal{A}$.
  Let $\Lambda_{g} = g \mathcal{O}_{\pi}^{3}$ and $\Lambda_{h} = h \mathcal{O}_{\pi}^{3}$ and $\Lambda_{h} \neq \Lambda_{g}$. 
  Thus $\Lambda_{g},\Lambda_{h}$ are distinct self-dual lattices in $F_{\pi}^{3}$.
  Then $\tilde{d}(g,h) = 2$ if and only if $ \Lambda_{g }\subseteq  \pi^{-1}\Lambda_{h}$ and $ \Lambda_{h} \subseteq \pi^{-1} \Lambda_{g}$.
\end{lemma}
\begin{proof}
    We know that $\tilde{d}(g,h)=2$ is possible if either $l(g^{-1}h) =2, l(h^{-1}g)=2$ or $l(g^{-1}h)=4,l(h^{-1}g) = 0$.

    Let us show that the latter cannot happen. If $l(h^{-1} g )=0$ then we must have $l(g)=l(h)$ since $ 0 = l(h^{-1}g) \geq \max\{ l(g) - l(h) , l(h)-l(g) \} \geq 0$. 
   Suppose $2n=l(g)$. 
   This means that for some $k_{1},k_{2},k_{1}',k_{2}' \in \GL_{3}(\mathcal{O}_{\pi})$ and $d = \diag(\chi^{n}, 1, \chi^{-n})$, by Lemma \ref{le:diagonals} we get 
   \begin{equation}
       g =  k_{1} d k_{1}', h = k_{2} d k_{2}'.
       \label{eq:cartan_gh}
   \end{equation}
   Denote $x = k_{1}^{-1} k_{2} \in \GL_{3}(\mathcal{O}_{\pi})$. Then, because of Lemma \ref{le:ispositive} we get that $l(g^{-1}h ) = l(d^{-1} x d) = 4$ and $l(h^{-1} g) = l(d^{-1} x^{-1} d) = l \left( (d^{-1} x d)^{-1}\right) = 0$. The last equality means $$d^{-1} x^{-1} d \in \left( d^{-1}\GL_{3}(\mathcal{O}_{\pi}) d \right) \cap \GL_{3}(\mathcal{O}_{\pi}).$$
   The right side is a congruence subgroup in $\GL_{3}(\mathcal{O}_{\pi})$ and therefore is preserved under taking inverses. This implies that $d^{-1} x d \in \GL_{3}(\mathcal{O}_{\pi})$ and therefore $l(d^{-1} x d) = 0 \neq 4$.

   Therefore, we conclude that $l(g^{-1}h) = l(h^{-1} g ) =2$. Let $ g^{-1} h= kd k'$ be the Cartan decomposition with $k,k' \in \GL_{3}(\mathcal{O}_{\pi})$ and the diagonal $d = \diag(\chi^{n_{1}}, \chi^{n_{2}}, \chi^{n_{3}})$. Because $l(g^{-1}h)=l(h^{-1}g)=2$, we know that $1 = n_{1} \geq n_{2} \geq n_{3} = -1$. This means that $\chi d$ and $\chi d^{-1}$ has all its entries in $\mathcal{O}_{\pi}$.
   This lets us conclude that $ g^{-1}  h \mathcal{O}_{\pi}^{3} = k d \mathcal{O}_{\pi}^{3} \supset k \pi \mathcal{O}_{\pi}^{3} = \pi \mathcal{O}_{\pi}^{3}$ and 
   therefore $ \Lambda_{g} \subseteq  \pi^{-1} \Lambda_{h}$. 
   Similarly, using $d^{-1}$ we conclude that $\Lambda_{h} \subseteq \pi^{-1} \Lambda_{g}$. 

   To show that $\Lambda_{h} \subseteq \pi^{-1} \Lambda_{g}$ and $\Lambda_{g} \subseteq \pi^{-1} \Lambda_{h}$ imply that $\tilde{d}(g,h) =2$, one can show that $l(g^{-1}h), l(h^{-1}g) = 2$. This follows in the same way as the above paragraph. 

\end{proof}
\section{Proof of Proposition \ref{prop:distance}}
\label{proofoflemma5.5}
\begin{proposition}
    Consider $\mathcal{B}$ as a graph and let $d(x,y)$ denote the edge-distance between two points $x,y \in \mathcal{B}$  
    Let $g,h \in \mathcal{A} \subseteq \GL_{3}(F_{\pi})$. 
    Then, one has $\tilde{d}(g,h) = d(g v_{0},hv_{0})$.
\end{proposition}

\begin{proof}
    
We know from Lemma \ref{le:Aset} that $\mathcal{A} v_{0} = \mathcal{P}_{\mathcal{B}}$.

To show that $\tilde{d}(g,h) = d(gv_{0}, hv_{0})$, it is sufficient to show that for any $g,h \in \mathcal{A}$
\begin{equation}
\tilde{d}(g,h) = 2 \text{  if and only if } d(g v_{0} , h v_{0}) = 2.
\end{equation}
Indeed, here's how this claim proves the result. For general $g,h \in \mathcal{A}$ such that $d(g v_{0} , h v_{0}) < \infty$, we obtain from the claim that $\tilde{d}(g,h)$ is at most 
\begin{align}
    2  & \min\{ n \mid (g_{0},\dots,g_{n}) \in \mathcal{A}^{n+1} ,
	g_{0}=g, g_{n} = h, 
    d(g_{i} v_{0},g_{i+1} v_{0}) = 2 \text{ for } i = 0,\dots,n-1
    \},
\end{align}
hence we get $\tilde{d}(g,h) \leq d(g v_{0}, h v_{0})$. 

On the other hand, if $d(g v_{0},h v_{0}) = 2n$ then one shows from constructing a chain of $1$-simplices that 
$ \Lambda_{g}  = g \mathcal{O}_{\pi}^{3}$ and $\Lambda_{h} = h \mathcal{O}_{\pi}^{3}$ satisfy $ \pi^{n} \Lambda_{h} \subseteq \Lambda_{g} \subseteq \pi^{-n}\Lambda_{h}$. This implies that $g^{-1} h \mathcal{O}_{\pi}^{3} \subseteq \pi^{-n} \mathcal{O}_{\pi}^{3}$ and therefore $l(g^{-1}h) \leq 2n$. Similarly we conclude the same for $l(h^{-1}g)$ and therefore obtain $\tilde{d}(g,h) \leq 2n$.

From the above discussion, it is immediately clear that $d(gv_{0},hv_{0}) =2$ implies that $\tilde{d}(g,h) = 2$. For the converse, we observe from Lemma \ref{le:hecke_nghbs} that if $\tilde{d}(g,h)=2$ then $ \pi \Lambda_{h} \subset \Lambda_{g} \subset \pi^{-1} \Lambda_{h}$. We need to show that there are two 1-simplices connecting $\{\pi^{i}\Lambda_{h}\}_{i \in \mathbb{Z}}$ and $\{\pi^{i} \Lambda_{g}\}_{i \in \mathbb{Z}}$. For this, we must show that there exists a lattice $\Lambda \subseteq F_{\pi}^{3}$ such that 
\begin{align}
 \pi \Lambda_{g}\subseteq \dual{\Lambda} \subseteq  \Lambda_{g} \subseteq \Lambda  \subseteq \pi^{-1} \Lambda_{g} \text{ and }
 \pi \Lambda_{h}\subseteq \dual{\Lambda} \subseteq  \Lambda_{h} \subseteq \Lambda  \subseteq \pi^{-1} \Lambda_{h}.
\end{align}
We observe that $\Lambda = \Lambda_{g}+ \Lambda_{h}$ or $\dual{(\Lambda_h \cap \Lambda_g)}$ works. One checks that 
\begin{equation}
\pi \Lambda_h,\pi \Lambda_g. \subseteq 
    \Lambda_{h} \cap \Lambda_{g}\subseteq \dual{(\Lambda_{g} + \Lambda_h)}  \subseteq \Lambda_{h},\Lambda_{g}\subseteq  {\Lambda_{g}+ \Lambda_{h}}  \subseteq 
    \dual{(\Lambda_{g} \cap \Lambda_h)} \subseteq \pi^{-1} \Lambda_h, \pi^{-1} \Lambda_g.
\end{equation} 
\end{proof}

\end{document}